\documentclass[draftcls,12pt,onecolumn]{IEEEtran}
\usepackage[utf8]{inputenc}
\usepackage[english]{babel}
\usepackage{amsthm}
\usepackage{amsfonts}%
\usepackage{amsmath}%
\usepackage{fixmath}%
\ifCLASSINFOpdf
\usepackage{caption}
\usepackage{subfloat}
\usepackage[pdftex]{graphicx}
\graphicspath{ {./plots/} }
\else
\fi
\usepackage{amsmath,amsfonts,amssymb,array}

\theoremstyle{definition}

\theoremstyle{remark}

\usepackage{algorithmic}
\usepackage{array}
\usepackage[table,xcdraw]{xcolor}
\usepackage{cite}
\newcommand{\E}{\mathbb{E}}
\usepackage{siunitx}
\allowdisplaybreaks

\usepackage{amsthm}
\usepackage{amsfonts}%
\usepackage{amsmath}%
\usepackage{fixmath}%
\usepackage{amssymb}%
\usepackage{graphicx}
\usepackage{comment}
\usepackage{url}
\usepackage{cite}
\usepackage{xcolor}
\usepackage{amsthm}
\usepackage{subfloat}
\usepackage{subfig}
\usepackage{flushend}

\newtheorem{theorem}{Theorem}

\begin{document}
		\title{How much Training is Needed in Downlink Cell-Free mMIMO under LoS/NLoS channels?}
			\author{\IEEEauthorblockN{Sai Manikanta Rishi Rani, Ribhu Chopra, Kumar Appaiah} 
		
		\thanks{S.M.R. Rani was with the Department of Electronics and Electrical Engineering, Indian Institute of Technology Guwahati, Assam, India for the duration of this work, he is now with the Department of Electrical and Computer Engineering, University of California San Diego, San Diego, CA 92093 USA. R. Chopra is with the Department of Electronics and Electrical Engineering, Indian Institute of Technology Guwahati, Assam, India. K. Appaiah is with the department of Electrical Engineering, Indian Institute of Technology Bombay, Mumbai, Maharashtra, India. (emails: sairishi10@gmail.com, ribhu@outlook.com,akumar@ee.iitb.ac.in).
		
		}
}

\maketitle		
\begin{abstract}
The assumption that no LoS channels exist between wireless access points~(APs) and user equipments~(UEs) becomes questionable in the context of the recent developments in the direction of cell free massive multiple input multiple output MIMO~(CF-mMIMO) systems. In CF-mMIMO systems, the access point density is assumed to be comparable to, or much larger than the the user density, thereby leading to the possibility of existence of LoS links between the UEs and the APs, depending on the local propagation conditions. In this paper, we compare the rates achievable by CF-mMIMO systems under probabilistic LoS/ NLos channels, with and without acquiring the channel state information~(CSI) of the fast fading components. We show that, under sufficiently large AP densities, statistical beamforming that does not require the knowledge about the fast fading components of the channels, performs almost at par with full beamforming, utilizing the information about the fast fading channel coefficients, thus potentially avoiding the need for training during every frame. We validate our results via detailed Monte Carlo simulations, and also elaborate the conditions under which statistical beamforming can be successfully employed in massive MIMO systems with LoS/ NLoS channels.

\end{abstract}
	
\section{Introduction}	
\subsection{Motivation}
The idea of using a large number of co-located antennas to simultaneously serve a smaller number of user equipments~(UEs) over the same time frequency resource, dubbed massive multiple input multiple output~(mMIMO), has been of much research interest over the past decade~\cite{Marzetta_TWC_2010,Chokes_BSR_LMIMO,Red_book}. Greater spectral and energy efficiencies coupled with analytical tractability and simple linear processing are some of the key benefits of this architecture~\cite{scaling,CHEMP}. Consequently, mMIMO has been widely accepted as a key enabling technology for 5G cellular systems~\cite{reality}, with a large body of literature discussing the  requirement of the availability of accurate channel state information~(CSI) at the BS. However, concerns have been raised on the non-uniform coverage offered by massive MIMO systems, that is inherently unfair to UEs located at the cell edges. The focus has therefore shifted to providing a uniform quality of service to all users in an mMIMO like setup. As a consequence, several distributed realizations of the mMIMO setup have recently been explored, with cell free~(CF)-mMIMO, emerging as a front-runner~\cite{Ngo_SPAWC_2015,Ngo_TWC_2017}. 

The canonical CF-mMIMO setup employs a large number of access
points~(APs), each equipped with a small antenna array or mostly a
single antenna, distributed over a large area and simultaneously
serving a large number of users over the same time frequency
resource. Here, the AP antenna density is assumed to be much larger
than the UE density, with all the APs connected to a central
processing unit~(CPU) via a high rate, loss free backbone link. It has
been shown that CF-mMIMO systems, in addition to offering a more
uniform coverage to users, also inherit several advantages of
conventional cellular mMIMO systems, such as power scaling and
simple linear processing~\cite{Ngo_SPAWC_2015}. Another interesting
aspect of CF-mMIMO systems is the distributed architecture that
results in close proximity between the users and the serving APs. This
also increases the probability of having a line of sight~(LoS) link
between a user and one or more of the associated APs. Most of the
existing literature on CF-mMIMO systems assumes either Rayleigh fading
channels with an LoS link existing with probability
0~\cite{Ngo_TWC_2017,Ngo_TGCN_2018,Hardening_CF,blessing,Shamai_VTC_2001,Nayebi_TWC_2017,making_CF,Bashar_ICC_2018,Dhillon_CFMM,Zhang_Access_2018,Papa_TVT_2020},
or Rician fading channels with an LoS link existing with probability
1~\cite{Ozdogan_TWC_2019,Zhang_Comml_2002,Jin_sys_2020}. However, due
to the uncertainties in the physical distribution of various blockages
(e.g. buildings in urban areas, trees, hills, etc. in rural areas), it
is impossible to accurately determine the presence of an LoS link
between the APs and the UEs apriory. Therefore, in this paper, we discuss the
performance of CF-mMIMO systems under probabilistic LoS
components. 

\subsection{Prior Work}
The idea of CF-mMIMO was first discussed
in~\cite{Ngo_TWC_2017,Ngo_TGCN_2018}, and it was shown that, even with
local conjugate beamforming, it can outperform both colocated mMIMO as
well as the small cell networks model, in terms of throughput and
energy efficiency. The idea of using only locally available channel
state information~(CSI) to serve a large number of users in CF-mMIMO
is based on the earlier idea of cooperative
MIMO~\cite{Shamai_VTC_2001}, wherein the basic idea is to group
multiple devices via point-to-point links into a virtual antenna array
to emulate MIMO communications.  In~\cite{Ngo_TWC_2017}, the
impracticality of forwarding the CSI from all APs to the CPU was
considered, and CPU data detection based on preprocessed samples from
the APs was suggested. Alternatively, the authors
in~\cite{Nayebi_TWC_2017} analyzed the performance of CF-mMIMO systems
with centralized MMSE combining under the assumption of CSI
availability at the CPU. The issue of impracticality of an infinite
rate backhaul was raised in~\cite{Bashar_ICC_2018}
and~\cite{Dhillon_CFMM}, where the effects of quantization in the
backhaul link were derived, and corresponding impairments to the
achievable rate were quantified. The authors in~\cite{making_CF} make
a detailed comparison of the backhaul requirements, communication and
computational complexities, and the achievable performances of
different combining schemes in CF-mMIMO systems. It has been
demonstrated that the advantage of channel hardening enjoyed by cellular mMIMO systems is also
inherited by CF-mMIMO systems~\cite{Ngo_TWC_2017}. However, it was
argued in~\cite{Hardening_CF} that, due to the distributed nature of
the system, the channels between different APs and users are
independent but non-identically distributed. This leads to CF massive
MIMO inheriting this property only when the number of antennas at
individual APs is large. 

In recent years, considerable research efforts have gone into
characterizing the behaviour of MIMO systems under LoS/NLoS
channels~\cite{Ding_TWC_2016,Ding_TWC_2017,Cho_Comml_2018}. We note
that the effect of probabilistic LoS/NLoS channels is more pronounced
in dense cellular networks than sparse networks due to the increased
AP and UE
densities. 
Depending on the degree of accuracy of the modelling for LoS/NLoS
channels, the behaviour of the area spectral efficiency (ASE) of the
system has also been observed to
change~\cite{Ding_TWC_2016,Ding_TWC_2017}. In~\cite{LoS_Uplink} the
uplink performance of CF-mMIMO systems under probabilistic LoS/NLoS
channels was analysed for different combining schemes being used at
the CPU, and it was found that, for sufficiently large AP densities,
ubiquitous LoS link coverage can be assumed for all the UEs. It has
recently been argued that, in case LoS links are present between the
BS and the UEs, the downlink beamforming based purely on the
statistical or slow fading parameters of the channel model, dubbed
statistical beamforming, can be
used~\cite{spatial_div,two_tier,3d_BF}. In~\cite{spatial_div,
  two_tier}, the second order statistics of the channel were exploited
to create a two tier precoding matrix. Since statistical CSI varies
slowly as compared to the real-time instantaneous CSI, it can
therefore be obtained accurately by the BS/CPU via long-term
feedback~\cite{stale_CSI}. Moreover, during the transmission, the UEs
only need to estimate and feedback the dimension-reduced effective
channel. This greatly reduces the channel estimation and feedback
overhead in mMIMO systems. The authors in~\cite{3d_BF} perform LoS
based beamforming on a mMIMO downlink system under a Rician fading
channel. The proposed statistical three-dimensional (3D) beamforming
precoding scheme is based only on the statistical LOS information of
each user. The scheme achieves considerably high sum rates while
requiring much less CSI overhead at the
BS. This results in an interesting proposition for CF-mMIMO systems,
where the cost of CSI acquisition at the CPU is high, and the
probability of having an LoS channel between a user and one or more
APs has been shown to be large~\cite{LoS_Uplink}.


\subsection{Contributions}
 
In this paper, we examine the need for the availability of fast fading
CSI in a CF-mMIMO system under a probabilistic LoS/NLoS channel. We
use the instantaneous achievable rate per channel use as a metric to
quantify the effect of the availability of the fast fading CSI at the
APs. In this discussion, we restrict ourselves to using conjugate
beamforming for downlink data precoding. This is mainly due to the
simplicity of implementation, and the reduced reliance on the backhaul
network due to its distributed implementation. For simplicity of
exposition, we dub schemes with and without the availability of fast
fading CSI as full and statistical beamforming schemes,
respectively. We summarize our main contributions as follows:

\begin{enumerate}
	
	\item  To benchmark the performance of all the other schemes we first derive bounds on the rate achievable by the system under the availability of accurate CSI at both the APs and the UEs~(see Section~\ref{subsec:Accurate}.). 
	
	\item Following this, we analyze the achievable rates of the system, with the statistical CSI being ubiquitously available, but the fast fading CSI being obtained via pilot training at both the APs and the UEs. We note that the estimation errors in the fast fading components lead to the introduction of self interference impairing the system performance~(see Section~\ref{subsec:trained}.).
	
	\item Subsequently, we derive theoretical bounds on the system performance under pure statistical beamforming, and with no information being available about the fast fading channel components~(see Section~\ref{subsec:statistical}.). 
	
	\item We then obtain bounds on the performance achievable with statistical beamforming, but aided by downlink pilot training. This is seen to noticeably reduce the inter-stream interference caused by the fast fading channel components~(see Section~\ref{subsec:dtrain}.).
	
	
	\item Finally, via detailed simulations, we validate the derived theory and compare the performance of the four precoding schemes, and draw conclusions on the power and performance trade offs involved in LoS reliant CF-mMIMO systems~(see Section~\ref{sec:results}.).
	
\end{enumerate}

The key takeaway of this work is that, in the case of CF-mMIMO
systems, with a sufficiently large AP density, transmission schemes
that are designed to rely more on the LoS component of the channels
perform significantly well; almost at par with schemes that utilize
the full CSI for beamforming. This implies that having accurate and up
to date CSI of the fast fading channel components becomes unnecessary
in CF-mMIMO systems that possess sufficiently large AP densities,
thereby allowing us to skip pilot based CSI estimation in each
coherence interval. We next describe the system model, including the
system setup, the channel model, the LoS link probability model, and
the power control and channel estimation schemes used in this paper.

\section{System Model}

We consider a CF-mMIMO system with $M$ APs that operate in the time
division duplex~(TDD) mode, each equipped with a uniform linear array
of $N$ antennas, serving a total of $K \ll MN$ UEs. The height of the
$m $th AP is assumed to be $\ell_m$, with its antennas separated by a
distance $d$. Similarly, the $k$th UE's height is assumed to be
$\ell_k^{\prime}$.
\subsection{Channel Model}
We note that the line-of-sight~(LoS) path between a UE and an AP may
be obstructed due to the presence of blockages (e.g. buildings in
urban areas, trees/hillocks in rural areas, etc.). However, due to the
random nature of the locations of the UEs, APs, and these blockages,
the existence of LoS paths between these cannot be
guaranteed. Therefore, in this work, we characterize the channel
between a UE and an AP as a combination of both LoS and
NLoS~(non-line-of-sight) channels. As a result, the channel between
the $m$th AP and the $k$th UE,
$\mathbf{h}_{mk}\in\mathbb{C}^{N\times1}$, is given by
\begin{equation}
\mathbf{h}_{mk}=\alpha_{mk}\dot{\mathbf{h}}_{mk}+\sqrt{\beta_{mk}}\mathbf{\bar{h}}_{mk}.
\label{eq:channelcombo}
\end{equation}
\noindent Here $\mathbf{\bar{h}}_{mk}$ represents the fast fading
component of the NLoS channel that is assumed to consist of
independent and identically distributed~(i.i.d.) zero mean circularly
symmetric complex Gaussian (ZMCSCG) entries with unit variance,
i.e. $\mathbf{\bar{h}}_{mk} \sim \mathcal{C}\mathcal{N}(0, \textbf
I_N)$. The coefficient $\beta_{mk}$ denotes the slow fading path loss
component of the NLoS channel. Also, $\alpha_{mk}$ is a Bernoulli
random variable that indicates the presence or the absence of an LoS
component between the $m$th AP and the $k$th UE, such that
$\Pr\{\alpha_{mk}=1\}=q_{mk}$ (a detailed discussion on the modeling of $q_{mk}$ follows in Section~{\ref{sec:IIA}}).
Finally, $\dot{\mathbf{h}}_{mk}$ denotes the LoS channel, and is given as
\begin{equation}
\dot{\mathbf{h}}_{mk}=\mathbf{a}(\theta_{mk})\sqrt{{G_mG_k}}\left(\frac{\ell_{k}^{'}\ell_m}{4\pi x_{mk}}\right)e^{\iota2\pi\frac{x_{mk}}{\lambda_c}},
\end{equation} 
with $x_{mk}$ being the three dimensional link distance between the $k$th UE and the $m$th AP, $\lambda_c$ denoting the carrier wavelength, $\iota\triangleq\sqrt{-1}$, $G_m$ and $G_k$ being the gains associated with the antennas of the $m$th AP and $k$th UE respectively, and $\mathbf{a}(\theta_{mk})$ representing the array response vector for a signal leaving from the $m$th AP towards the $k$th UE at an angle $\theta_{mk}$, given as 
\begin{equation}
\label{eq:aoavec}
\mathbf{a}(\theta_{mk})=\left[1,e^{\iota2\pi\frac{d}{\lambda_c}\sin(\theta_{mk})},e^{\iota4\pi\frac{d}{\lambda_c}\sin(\theta_{mk})},\ldots,e^{\iota2(N-1)\pi\frac{d}{\lambda_c}\sin(\theta_{mk})}\right]^T.
\end{equation}
The downlink frame consists of three sub-frames, viz. uplink training comprising at most $K$ channel uses for acquisition of CSI at the APs, downlink training comprising at most $K$ channel uses for acquisition of the effective downlink channels at the UEs, and actual data transmission. We next describe the LoS link probability model used in this work. 
\subsection{LoS link probability model}~\label{sec:IIA}
We note that the channel between the $m$-th AP and the $k$-th UE may or may not contain an LoS component, depending on the existence of blockages between them, and the heights of these blockages. Therefore, the existence of an LoS link depends on geographical distribution parameters of the network, such as the building height distribution, building locations, heights of APs and UEs, etc. The existing 3GPP model~\cite{3GPPr112_2}, however, does not account for all these factors, and hence cannot characterize the exact LoS link probabilities. A more complete model for the LoS link probabilities has been discussed in~\cite{Atzeni_TWC_2018} using the ITU blockage model~\cite{Kim_TWC_2020}. Using this model, the LoS link probability $q_{mk}$ is given by
\begin{equation}
	\label{eq:losprob}
	q_{mk} =  (1 - \omega)^{\sqrt{\eta \, \mu \, d_{mk}}} 
\end{equation}
where $\omega \triangleq \sqrt{\frac{\pi}{2}}\frac{\rho}{\ell_m –
  \ell^{\prime}_k} \left[
  \text{erf}\left(\frac{\ell_m}{\rho\sqrt{2}}\right) -
  \text{erf}\left(\frac{\ell^{\prime}_k}{\rho\sqrt{2}}\right)\right]$,
and $\text{erf}(z) \triangleq
\frac{1}{\sqrt{\pi}}\int_{-z}^{z}e^{-t^2} dt$, $\rho$ is the average altitude of blockages,  $\eta$ is the fraction of the built up area, and $\mu$ is the average number of blockages per unit area~\cite{IMT2020propagate}. This completes the description of the channel model. We next describe the uplink channel estimation procedure. 

\subsection{Uplink Channel Estimation}
\label{sec:IIB}

For simplicity, we assume the availability of $K$ orthogonal uplink pilot sequences in each coherence interval, with the pilot symbol sent by the $k$th UE at the $n$th instant given as $\psi_k[n]$, such that $\sum_{n=1}^{T}\psi_k[n]\psi_{l}^*[n]=\delta[k-l]$, where $\delta[\cdot]$ denotes the Kronecker delta. We also assume that all the users share the same average pilot power that we can normalize to unity without the loss of generality, and the corresponding normalized noise power being given by $\sigma^2_u$. Therefore, the pilot signal received by the $m$th AP at the $n$th instant is given by
\begin{equation}
\mathbf{y}_m[n]=\sum_{k=1}^K \mathbf{h}_{mk}\psi_{k}[n]+\sigma_u\mathbf{w}_k[n].
\end{equation}
where $\mathbf{w}_k[n] \sim \mathcal{CN}(0,\textbf I_N)$ denotes the independent additive white Gaussian noise (AWGN). Defining $\mathbf{y}^{'}_{mk}=\left(\sum_{n=1}^{T} \mathbf{y}_{m}[n]\psi_{k}^*[n]\right) - \alpha_{mk}\mathbf{\dot h}_{mk}$, we obtain
\begin{equation}
\mathbf{y}^{'}_{mk}=  \sqrt{\beta_{mk}}\mathbf{\bar h}_{mk}  +\sigma_u\mathbf{w}^{'}_k[n].
\label{eq:processpilot}
\end{equation}
where $\mathbf{w}^{'}_k[n]$ is defined implicitly. This is done under
the assumption that the APs have an accurate estimate of the LoS CSI through long-term feedback techniques~\cite{3d_BF}. We note that it is common to assume the exact knowledge of the deterministic LoS component of $\mathbf{h}_{mk}$, and consider only the estimate of $\mathbf{\bar{h}}_{mk}$~\cite{stale_CSI}.

Now,
$
 \E \left[\mathbf{\bar h}_{mk}\mathbf{\bar h}^H_{mk}\right]=\mathbf{I}_N, 
$
and consequently,
$
 \E \left[\mathbf{y'}_{mk}\mathbf{y'}^H_{mk} \right]=( \beta_{mk} +\sigma^2_u)\mathbf{I}_N,
$
such that,
$
 \E \left[\mathbf{\bar h}_{mk}\mathbf{y'}^H_{mk} \right]=\sqrt{\beta_{mk}}\mathbf{I}_N.
$
Based on this, and \eqref{eq:processpilot}, we can now write the LMMSE estimate $\mathbf{\hat{h}}_{mk}$ of $\mathbf{h}_{mk}$ as 
\begin{equation}
\mathbf{\hat{h}}_{mk}=\frac{\sqrt{\beta_{mk}}} { \beta_{mk} +\sigma^2_u} \mathbf{y}^{'}_{mk}.
\end{equation}
We can now express the NLoS component of the channel in terms of its estimate, and the associated estimation error as,
\begin{equation}
\mathbf{\bar h}_{mk}=\mathbf{\hat{h}}_{mk}+\mathbf{\tilde{h}}_{mk},
\label{eq:ch_est}
\end{equation}
with $\mathbf{\tilde{h}}_{mk}$ being a ZMCSCG random vector with $\E [ \mathbf{\tilde{h}}_{mk} \mathbf{\tilde{h}}_{mk}^T  ] = \frac{\sigma^2_u} { \beta_{mk} +\sigma^2_u} \mathbf{I}_N$, and $\\E [\mathbf{\tilde{h}}_{mk}\mathbf{\hat{h}}^H_{mk}]=\mathbf{O}_N$, where $\mathbf{O}_N$ corresponds to the all zero, order $N$ square matrix.
At this point, it is also convenient to define the overall channel estimate available at the APs as 
$
\mathbf{\check h}_{mk}\triangleq\mathbf{\dot h}_{mk} + \mathbf{\hat{h}}_{mk}.
$
The APs then use the available channel estimates to precode the downlink data, as described in the next subsection.

\subsection{Downlink Signal Model and Power Control}
Letting $s_k$ be the symbol being sent to the $k$th user, such that $\mathbf{s}=[s_1,s_2,\ldots,s_K]^T$, and $\mathbf{\check{H}}_m=  \mathbf{\dot{H}}_m + \mathbf{\hat{H}}_m$ being the channel estimate available at the $m$th AP, we can write the downlink signal vector transmitted by the $m$th AP as
    $$\mathbf{y}_m = \mathbf{\check H}_m ^* \text{diag}(\mathbf{x}_m) \mathbf{s} = \textbf{P}_m \textbf{s},$$
    with the $i$th entry, $x_{mi}$ of $\mathbf{x}_{m}$ corresponding to the power allocated by the $m$th AP to the $i$th user's channel, and $\mathbf{P}_m$ being the effective precoder matrix at the $m$th AP. Also, $\textbf{y}_m$ can be seen as consisting of $K$ components corresponding to the $K$ users, such that $\textbf{y}_{mk}= \textbf{p}_{mk}s_k$, with $\textbf{p}_{mk}$ being the $k$th column of $\textbf{P}_m$.
    
    Now, the signal received at the $k$th UE, $r_k$, can be written as,
\begin{equation}
    r_k= \sum_{m=1}^M \mathbf{h}_{mk}^T \mathbf{y}_m +\sigma_o w_k = \left( \sum_{m=1}^M \mathbf{h}_{mk}^T \mathbf{P}_m \right)\mathbf{s} + \sigma_o w_k
\end{equation}
where $w_k$ is the additive white Gaussian noise with zero mean and unit variance, and $\sigma_o^2$ is the normalized noise power.
We now define the composite channel vector from all the APs the to the $k$th user, as $\boldsymbol{\gamma}_k$ whose $i$th element $\gamma_{ki}$ is the effective channel coefficient between the $k$th UE and the downlink data stream targeted towards the $i$th user, such that,
$
   \boldsymbol{\gamma}_k = \left(\sum_{m=1}^M \mathbf{h}_{mk}^T \mathbf{P}_m\right)^T, 
$
and
$
   \gamma_{ki} = \sum_{m=1}^M \mathbf{h}_{mk}^T \mathbf{p}_{mi}.
$
Also, since $ \mathbf{p}_{mi} = x_{mi} \mathbf{\check h}^*_{mi} $, $\gamma_{ki} = \sum_{m=1}^M  x_{mi} \mathbf{h}_{mk}^T  \mathbf{\check h}^*_{mi}$, where $\mathbf{\check h}^*_{mi}= \mathbf{\dot h}^*_{mi}+ \mathbf{\hat h}^*_{mi}$.
    
It is important to note that $\gamma_{kk}$ can be expressed as $\gamma_{kk}=\dot\gamma_{kk} + \bar\gamma_{kk} $, where
\begin{equation}
    \dot\gamma_{kk} = \sum_{m=1} ^M x_{mk}  \alpha_{mk} \mathbf{\dot h}_{mk}^T \mathbf{\dot h}_{mk}^*,
\end{equation}
and $\bar \gamma_{kk}$ is expanded as 
\begin{equation}
    \bar\gamma_{kk} = \sum_{m=1} ^M x_{mk} \left( \alpha_{mk} \left(  \mathbf{\dot h}_{mk}^T \mathbf{\hat h}_{mk}^* + \mathbf{\bar h}_{mk}^T\mathbf{\dot h}_{mk}^*   \right) + \mathbf{\bar h}_{mk}^T\mathbf{\hat h}_{mk}^* \right).
\end{equation}
The term $\dot\gamma_{kk}$ is the component of $\gamma_{kk}$ that consists of only statistical and slow fading terms that can be safely assumed to be known at the UEs, whereas $\bar\gamma_{kk}$ is composed of the fast fading channel coefficients. 
    
We consider two forms of downlink power control, viz. fixing the power transmitted by each AP, i.e. $\E[|\mathbf{y}_m \|^2]=\mathcal{E}_{\text{a,max}}$, and fixing the power expended for each UE, i.e. $\sum_{m=1}^M \E[\| y_{mk} \|^2]=\mathcal{E}_{\text{u,max}}$.

\subsection{Downlink LMMSE estimation of effective channel vector}
\label{sec:IID}

We again assume the availability of $K$ orthogonal downlink pilot sequences in each coherence interval, with the pilot symbol being sent to the $k$th UE at the $n$th instant given as $\xi_k[n]$, such that $\sum_{n=1}^{T}\xi_k[n]\xi_{l}^*[n]=\delta[k-l]$. We also assume $\sigma^2_d$ to be the noise power normalized by the downlink pilot power. Therefore, the pilot signal received by the $k$th AP at the $n$th instant is given by
\begin{equation}
y_k[n]=\sum_{i=1}^K \gamma_{ki}\xi_{i}[n]+\sigma_d w_k[n].
\end{equation}
\noindent where $w_k[n] \sim \mathcal{CN}(0, 1)$ denotes the independent additive white Gaussian noise (AWGN). Defining $y^{'}_{kl}=\sum_{n=1}^{T}y_{k}[n]\xi_{l}^*[n]$, we obtain
\begin{equation}
y^{'}_{kk}= \gamma_{kk}+\sigma_d w^{'}_k.
\end{equation}
Now letting
\begin{equation}
\bar y_{kk}= y^{'}_{kk} - \dot\gamma_{kk} = \bar\gamma_{kk} +\sigma_d w^{'}_k,
\end{equation}
we obtain
$
\E \left[  \bar y_{kk} \bar y_{kk}^*  \right] = \E \left[ \lvert \bar\gamma_{kk}\rvert^2 \right] + \sigma^2_d, 
$ 
$
 \E \left[ \bar \gamma_{kk}  \bar y_{kk}^*   \right] = \E \left[ \lvert \bar\gamma_{kk}\rvert^2 \right], 
$
and based on this, we can write the LMMSE estimate $\hat \gamma_{kk}$  of $\bar \gamma_{kk}$ as, 
\begin{equation}
\hat \gamma_{kk} =  \E \left[ \bar \gamma_{kk}  \bar y_{kk}^*   \right]  \E \left[  \bar y_{kk} \bar y_{kk}^*  \right]^{-1} \bar y_{kk}= \frac{\E \left[ \lvert \bar\gamma_{kk}\rvert^2 \right] \bar y_{kk}} {\E \left[ \lvert \bar\gamma_{kk}\rvert^2 \right] + \sigma^2_d} 
\end{equation}
Therefore, 
\begin{equation}
    \bar\gamma_{kk} =  \hat\gamma_{kk} + \Tilde\gamma_{kk}.
\end{equation}
with
   $\E [ \hat\gamma_{kk} \Tilde\gamma_{kk}^*] = 0,$ and
\begin{equation}
\label{eqn:gam_hat}
   \E \left[ \lvert \hat\gamma_{kk}\rvert^2 \right] = \frac{\left(\E \left[ \lvert \bar\gamma_{kk}\rvert^2 \right] \right)^2} {\E \left[ \lvert \bar\gamma_{kk}\rvert^2 \right] + \sigma^2_d} .
\end{equation}
%
Furthermore,
$  \E \left[ \lvert \bar\gamma_{kk}\rvert^2 \right] =  \E \left[ \lvert \hat\gamma_{kk}\rvert^2 \right] +  \E \left[ \lvert \tilde\gamma_{kk}\rvert^2 \right] $,
and similarly,
$
   \E \left[ \lvert \gamma_{kk}\rvert^2 \right] =  \E \left[ \lvert \dot\gamma_{kk}\rvert^2 \right] + \E \left[ \lvert \hat\gamma_{kk}\rvert^2 \right] +  \E \left[ \lvert \tilde\gamma_{kk}\rvert^2 \right].
$

We next evaluate the performance of the system described above, under various full and statistical beamforming schemes.  

\section{Performance Under Full Beamforming}
In this section, we discuss the case where the APs use full beamforming, that is, the knowledge of both the LoS and the fast fading components to beamform the data to the different UEs. We first consider the availability of accurate CSI at the APs and the UEs to obtain an upper bound on the system performance. 
\subsection {Accurate CSI at APs and UEs}\label{subsec:Accurate}
In case accurate CSI is available at all the APs, the precoder matrix becomes,
$\mathbf{P}_m = \mathbf{H}_m ^* \text{diag}(\mathbf{x}_m) $, and the signal received at the $k$th UE, $r_k$, can be written as,
\begin{equation}
    r_k = \left( \sum_{m=1}^M \mathbf{h}_{mk}^T \mathbf{P}_m \right)\mathbf{s} + \sigma_o w_k = \boldsymbol{\gamma}_k^T \mathbf{s}  +\sigma_o w_k.
\end{equation}
Equivalently,
 \begin{equation}
    r_k= \gamma_{kk}s_k + \sum^K_{ \substack{i=1 \\ i\neq k}} \gamma_{ki} s_i +\sigma_o  w_k,
    \label{eq:perfect_recd_sig}
 \end{equation}
where the first term indicates the signal of interest, the second indicates the inter-stream interference, and the last term denotes the AWGN. Now based on this, the achievable data rate for the $k$th user can be upper bounded according to Theorem~\ref{thm:perfect}.

\begin{theorem}
	An upper bound on the rate achievable by the $k$th UE can be written as,
	\begin{equation}
    R_k \leq  \log\left( 1+ \frac{\E \left[\lvert \gamma_{kk}\rvert^2\right]}{\sum^K_{ \substack{i=1 \\ i\neq k}} \E \left[\lvert\gamma_{ki} \rvert ^2\right] + \sigma^2_o  }\right),
\end{equation}
where,
\begin{align}
\label{eq:stat_acc_kk}
\E \left[\lvert\gamma_{kk}\rvert ^2\right] = \sum _{m=1} ^M |x_{mk}|^2 \left(q_{mk} \zeta_{mkk} ^2 + N (N+1)\beta_{mk}^2 + 2 (N+1) q_{mk} \beta_{mk} \zeta_{mkk}\right) \nonumber
\\+\sum_{m=1}^M \sum_{\substack{m'=1\\m'\neq m}}^M x_{mk}x_{m'k}  \left( q_{mk} \zeta_{mkk} + \beta_{mk} N \right)\left( q_{m'k} \zeta_{m'kk} + \beta_{m'k} N \right)^*,
\end{align}
and
\begin{align}
\label{eq:stat_acc_ki}
\E \left[\lvert\gamma_{ki}\rvert ^2\right] = \sum _{m=1} ^M |x_{mi}|^2 \left(q_{mk}q_{mi}  \lvert \zeta_{mki} \rvert ^2  + \beta_{mk} \beta_{mi} N + q_{mk} \beta_{mi} \zeta_{mkk}
+ \beta_{mk} q_{mi} \zeta_{mii} \right) \nonumber
\\
+\sum_{m=1}^M \sum_{\substack{m'=1\\m'\neq m}}^M x_{mi}x_{m'i} q_{mk} q_{mi} q_{m'k} q_{m'i}  \zeta_{mki} \zeta_{m'ki}^*,
\end{align}
 with $q_{mk} \triangleq \Pr\{\alpha_{mk}=1\}$, and $\zeta_{mki} \triangleq \mathbf{\dot h}_{mk} ^T \mathbf{\dot h}_{mi}^*$.
\label{thm:perfect}
\end{theorem}

\begin{proof}

Treating the interference in~\eqref{eq:perfect_recd_sig} as noise~\cite{HandH}, and under the assumption that $\E \left[s_i s_j^*\right] = \delta\left[i-j\right]$, that is, the symbols sent to the different users are independent of each other, the rate achievable by the $k$th user, $R_k$, can be written as,
\begin{equation}
    R_k = \E \left[ \log\left(1+ \frac{\lvert \gamma_{kk}\rvert^2}{\sum^K_{ \substack{i=1 \\ i\neq k}} \E \left[\lvert\gamma_{ki} \rvert ^2 \right] + \sigma^2_o  }\right)\right].
    \label{eq:Rk_perfect}
\end{equation}

Invoking Jensen's inequality, we can upper bound~\eqref{eq:Rk_perfect} as,
\begin{equation}
    R_k \leq  \log\left( 1+ \frac{\E \left[\lvert \gamma_{kk}\rvert^2\right]}{\sum^K_{ \substack{i=1 \\ i\neq k}} \E \left[\lvert\gamma_{ki} \rvert ^2\right] + \sigma^2_o.  }\right).
\end{equation}
The statistics for the effective direct and interfering downlink channels have been derived in Appendix~\ref{app:stat_perfect}. 
\end{proof}

\subsection{Estimated CSI at APs and UEs}
\label{subsec:trained}

In this case, the precoding matrix consists of the estimated  uplink channel coefficients, as discussed in Section~\ref{sec:IIB}, and therefore, $\mathbf{P}_m =\left (\mathbf{\dot H}^*_m diag(\boldsymbol{\alpha_m}) + \mathbf{ \hat H}^*_m\right)\mathbf{X}_m = \mathbf{\check H}^{*}_m \mathbf{X}_m$.

Based on this, we can write the effective downlink channel for the $k$th user's data stream as, $\gamma_{kk} = \dot\gamma_{kk} + \bar\gamma_{kk}$, with the slow fading component $\dot\gamma_{kk}$ accurately known and the UEs and, $\bar\gamma_{kk}$ being estimated as $\hat\gamma_{kk}$, as described in Section~\ref{sec:IID}.

Therefore, the received signal at the $k$th user, $r_k$, can be expressed as,
 \begin{equation}
     r_k = \left( \dot\gamma_{kk} + \hat\gamma_{kk} \right)s_k+ \tilde\gamma_{kk}s_k + \sum^K_{ \substack{i=1 \\ i\neq k}} \gamma_{ki} s_i + \sigma_o w_k,
     \label{eq:est_recd_sig}
 \end{equation}
where the first term denotes the desired signal, the second term denotes the self interference due to channel estimation error, the third term denotes the inter-stream interference and the last term is the AWGN. Now based on this, the achievable data rate for the $k$th user can be upper bounded according to Theorem~\ref{thm:est}.

\begin{theorem}\normalfont
\label{thm:est}
	An upper bound on the rate achievable by the $k$th UE can be written as,
	\begin{equation}
    R_k \leq  \log\left( 1+ \frac{\E \left[\lvert \dot\gamma_{kk}\rvert^2\right] + \E \left[\lvert \hat\gamma_{kk}\rvert^2\right]}{ \E \left[ \lvert \tilde\gamma_{kk}   \rvert^2 \right] +\sum^K_{ \substack{i=1 \\ i\neq k}} \E \left[\lvert\gamma_{ki} \rvert ^2\right] + \sigma^2_o  }\right)
\end{equation}

where, 
\begin{multline}
\label{eq:est_acc_ki}
\E \left[\lvert\gamma_{ki}\rvert ^2\right] = \sum _{m=1} ^M |x_{mi}|^2 \Bigg(q_{mk}q_{mi} \lvert\zeta_{mki}\rvert^2  +N \beta_{mk} \beta_{mi} \frac{ \beta_{mi}}{\beta_{mi}+\sigma_u^2} + q_{mk} \beta_{mi} \frac { \beta_{mi} }{\beta_{mi}+\sigma_u^2}\zeta_{mkk} 
\\  + q_{mi} \beta_{mk} \zeta_{mii} \Bigg)
+\sum_{m=1}^M \sum_{\substack{m'=1\\m'\neq m}}^M x_{mi}x_{m'i} q_{mk} q_{mi} q_{m'k} q_{m'i}  \zeta_{mki} \zeta_{m'ki}^*,
\end{multline}

\begin{equation}
\label{eq:stat_dot}
    \E \left[\lvert \dot \gamma_{kk} \rvert ^2\right] =\sum_{m=1}^M \lvert x_{mk}\rvert^2 q_{mk} \zeta_{mkk}^2
    + \sum_{m=1}^M \sum_{\substack{m'=1\\m'\neq m}}^M  x_{mi}x_{m'i} q_{mk} q_{m'k}  \zeta_{mkk} \zeta_{m'kk}^*,
\end{equation}

\begin{align}
\label{eq:est_bar}
\E \left[\lvert\hat{\gamma}_{kk}\rvert ^2\right] = \sum _{m=1} ^M |x_{mk}|^2 \left(  N(N+1)\beta_{mk}^2 \left(\frac{\beta_{mk}}{\beta_{mk}+\sigma_u^2}\right)^2  +N \beta_{mk}^2 \left(\frac{\sqrt{\beta_{mk}} \sigma_u}{\beta_{mk}+\sigma_u^2}\right)^2 + \right.  
\\ \left. q_{mk} \beta_{mk} \left( 1 + \frac{\beta_{mk}}{\beta_{mk}+\sigma_u^2}\right) \zeta_{mkk} 
+ 2 N q_{mk}  \beta_{mk} \frac{\beta_{mk} }{\beta_{mk} + \sigma_u^2}  \zeta_{mkk}   \right) \nonumber
\\
+\sum_{m=1}^M \sum_{\substack{m'=1\\m'\neq m}}^M x_{mk}x_{m'k}  \left(q_{mk}\zeta_{mkk} N \beta_{m'k} \frac{ \beta_{m'k}}{\beta_{m'k}+ \sigma_u^2}  \nonumber
\right. \\ \left.
+ q_{m'k}\zeta_{m'kk}  N\beta_{mk} \frac{ \beta_{mk}}{\beta_{mk}+ \sigma_u^2} 
 + N^2\beta_{mk} \beta_{m'k} \frac{ \beta_{mk}}{\beta_{mk}+ \sigma_u^2} \frac{ \beta_{m'k}}{\beta_{m'k}+ \sigma_u^2}\right), \nonumber
\end{align}
 
$\E \left[ \lvert \hat\gamma_{kk}\rvert^2 \right] = \frac{\left(\E \left[ \lvert \bar\gamma_{kk}\rvert^2 \right] \right)^2} {\E \left[ \lvert \bar\gamma_{kk}\rvert^2 \right] + \sigma^2_d} ,$
and $\E \left[ \lvert \tilde\gamma_{kk}\rvert^2 \right] = \frac{\E \left[ \lvert \bar\gamma_{kk}\rvert^2 \right] \sigma^2_d } {\E \left[ \lvert \bar\gamma_{kk}\rvert^2 \right] + \sigma^2_d},$
 with $q_{mk} \triangleq \Pr\{\alpha_{mk}=1\}$, and $\zeta_{mki} \triangleq \mathbf{\dot h}_{mk} ^T \mathbf{\dot h}_{mi}^*$.
\end{theorem}

\begin{proof}

Treating the interference in~\eqref{eq:est_recd_sig} as noise, assuming $\E \left[s_i s_j^*\right] = \delta\left[i-j\right]$, the rate achievable by the $k$th user, $R_k$, can be written as,

\begin{equation}
\label{eq:Rk_est}
    R_k = \E \left[ \log\left(1+ \frac{\lvert \dot\gamma_{kk} + \hat\gamma_{kk}\rvert^2}{  \E \left[ \lvert \tilde\gamma_{kk}   \rvert^2 \right] + \sum^K_{ \substack{i=1 \\ i\neq k}} \E \left[ \lvert\gamma_{ki} \rvert ^2 \right] + \sigma^2_o  }\right)\right].
\end{equation}
Invoking Jensen's inequality we can upper bound~\eqref{eq:Rk_est} as,
\begin{equation}
    R_k \leq  \log\left( 1+ \frac{\E \left[\lvert \dot\gamma_{kk}\rvert^2\right] + \E \left[\lvert \hat\gamma_{kk}\rvert^2\right]}{ \E \left[ \lvert \tilde\gamma_{kk}   \rvert^2 \right] +\sum^K_{ \substack{i=1 \\ i\neq k}} \E \left[\lvert\gamma_{ki} \rvert ^2\right] + \sigma^2_o  }\right).
\end{equation}
The statistics for the effective direct and interfering downlink channels, have been derived in Appendix~\ref{app:stat_est}. 
\end{proof}

Having looked at the performance of the CF-mMIMO system under the availability of accurate/ estimated fast fading components, we next consider the case where no information about the fast fading components is available at the CPU/ APs. As a consequence, the APs use only the information about the LoS components to beamform the data to the respective UEs.

\section{Performance Analysis with Statistical Beamforming}

We now consider the case where the BS employs statistical beamforming in the downlink~\cite{3d_BF}. It is important to note that, even in the absence of any fast fading CSI at the UEs, the APs can still perform limited downlink training to allow the UEs to estimate the effective downlink channels, and use those to detect the received data symbols. As a consequence, we consider two variants of the statistical downlink beamforming scheme, viz. with and without downlink pilots. 
However, in either case, the precoding matrix is composed of only the LoS components of the channel matrix and is given as  $\mathbf{P}_m =\mathbf{\dot H}^*_m \text{diag}(\boldsymbol{\alpha}_m) \mathbf{X}_m $. 
Consequently, 
\begin{equation}
    \gamma_{ki} = \sum_{m=1} ^M x_{mi} \alpha_{mi} \mathbf{h}_{mk}^T \mathbf{\dot h}_{mi}^*,
\end{equation}
with $\gamma_{kk}=\dot{\gamma}_{kk}+\bar{\gamma}_{kk}$ such that,
$\dot\gamma_{kk} = \sum_{m=1} ^M x_{mk} \alpha_{mk}  \mathbf{\dot h}_{mk}^T \mathbf{\dot h}_{mk}^*, $
and
$
    \bar\gamma_{kk} = \sum_{m=1} ^M x_{mk} \alpha_{mk} \sqrt{\beta_{mk}} \mathbf{\bar h}_{mk}^T \mathbf{\dot h}_{mk}^*.  
$
We first consider the case where no downlink pilots are transmitted by the APs.
\subsection{No Downlink Pilots}
\label{subsec:statistical}
We note that, 
 \begin{equation}
 \label{eq:stat_recd_sig}
     r_k= \dot\gamma_{kk}s_k+ \bar\gamma_{kk}s_k + \sum^K_{ \substack{i=1 \\ i\neq k}} \gamma_{ki} s_i + \sigma_o w_k,
 \end{equation}
with the first term denoting the desired signal, the second term denoting the self interference due to the fast fading component of the channel, the third term denoting the inter-stream interference and the last term being the AWGN. Now based on this, the achievable data rate for the $k$th user can be upper bounded according to Theorem~\ref{thm:stat}.

\begin{theorem}\normalfont
\label{thm:stat}
	An upper bound on the rate achievable by the $k$th UE can be written as,
	\begin{equation}
    R_k \leq  \log\left( 1+ \frac{\E \left[\lvert \dot\gamma_{kk}\rvert^2\right]}{ \E \left[ \lvert \bar\gamma_{kk}   \rvert^2 \right] +\sum^K_{ \substack{i=1 \\ i\neq k}} \E \left[\lvert\gamma_{ki} \rvert ^2\right] + \sigma^2_o  }\right)
\end{equation}

where,

\begin{equation}
\begin{split}
\label{eq:stat_ki1}
    \E \left[\lvert \gamma_{ki} \rvert ^2\right] =\sum_{m=1}^M \lvert x_{mi}\rvert^2 \left(  q_{mk} q_{mi} \lvert \zeta_{mki} \rvert ^2 + \beta_{mk} q_{mi} \zeta_{mii} \right)
    \\+ \sum_{m=1}^M \sum_{\substack{m'=1\\m'\neq m}}^M  x_{mi}x_{m'i} q_{mk} q_{mi} q_{m'k} q_{m'i}  \zeta_{mki} \zeta_{m'ki}^*,
\end{split}
\end{equation}

\begin{equation}
\label{eq:stat_dot1}
    \E \left[\lvert \dot \gamma_{kk} \rvert ^2\right] =\sum_{m=1}^M \lvert x_{mk}\rvert^2 q_{mk} \zeta_{mkk}^2
    + \sum_{m=1}^M \sum_{\substack{m'=1\\m'\neq m}}^M  x_{mi}x_{m'i} q_{mk} q_{m'k}  \zeta_{mkk} \zeta_{m'kk}^*,
\end{equation}
and
\begin{equation}
 \E \left[ \lvert \bar{\gamma}_{kk}   \rvert^2 \right] =\sum_{m=1}^M \lvert x_{mk}\rvert^2\beta_{mk} q_{mk} \zeta_{mkk},
\end{equation}
with $q_{mk} \triangleq \Pr\{\alpha_{mk}=1\}$, and $\zeta_{mki} \triangleq \mathbf{\dot h}_{mk} ^T \mathbf{\dot h}_{mi}^*$.
	
\end{theorem} 

\begin{proof}

Treating the interference in~\eqref{eq:stat_recd_sig} as noise, assuming $\E \left[s_i s_j^*\right] = \delta\left[i-j\right]$, we can write,
\begin{equation}
\label{eq:Rk_stat}
    R_k = \E \left[ \log\left(1+ \frac{\lvert \dot\gamma_{kk}\rvert^2}{  \E \left[ \lvert \bar\gamma_{kk}   \rvert^2 \right] + \sum^K_{ \substack{i=1 \\ i\neq k}} \E \left[ \lvert\gamma_{ki} \rvert ^2 \right] + \sigma^2_o  }\right)\right].
\end{equation}
Invoking Jensen's inequality, we can upper bound~\eqref{eq:Rk_stat} as
\begin{equation}
    R_k \leq  \log\left( 1+ \frac{\E \left[\lvert \dot\gamma_{kk}\rvert^2\right]}{ \E \left[ \lvert \bar\gamma_{kk}   \rvert^2 \right] +\sum^K_{ \substack{i=1 \\ i\neq k}} \E \left[\lvert\gamma_{ki} \rvert ^2\right] + \sigma^2_o  }\right).
\end{equation}
The statistics for the effective direct and interfering downlink channels, have been derived in Appendix~\ref{app:stat_stat}. 

\end{proof}

\subsection{With Limited Downlink Training}
\label{subsec:dtrain}
 Using limited downlink pilots, we can estimate the fast fading component of the effective downlink channel to the $k$th user, $\bar\gamma_{kk}$ as $\hat\gamma_{kk}$. We note that the received signal can be expressed as
 \begin{equation}
 \label{eq:stat2_recd_sig}
     r_k= \left( \dot\gamma_{kk} + \hat\gamma_{kk} \right)s_k+ \tilde\gamma_{kk}s_k + \sum^K_{ \substack{i=1 \\ i\neq k}} \gamma_{ki} s_i + \sigma_o w_k,
 \end{equation}
with the first term denoting the desired signal, the second term denoting the self interference due to the error in estimating the fast fading component of the channel, the third term denoting the inter-stream interference and the last term being the AWGN. Based on this, the achievable data rate for the $k$th user can be upper bounded according to Theorem~\ref{thm:stat_2}.

\begin{theorem}\normalfont
\label{thm:stat_2}
	An upper bound on the rate achievable by the $k$th UE can be written as,
	\begin{equation}
    R_k \leq  \log\left( 1+ \frac{\E \left[\lvert \dot\gamma_{kk}\rvert^2\right] + \E \left[\lvert \hat\gamma_{kk}\rvert^2\right]}{ \E \left[ \lvert \tilde\gamma_{kk}   \rvert^2 \right] +\sum^K_{ \substack{i=1 \\ i\neq k}} \E \left[\lvert\gamma_{ki} \rvert ^2\right] + \sigma^2_o  }\right)
\end{equation}

where,

\begin{equation}
\begin{split}
\label{eq:stat_ki2}
    \E \left[\lvert \gamma_{ki} \rvert ^2\right] =\sum_{m=1}^M \lvert x_{mi}\rvert^2 \left(  q_{mk} q_{mi} \lvert \zeta_{mki} \rvert ^2 + \beta_{mk} q_{mi} \zeta_{mii} \right)
    \\+ \sum_{m=1}^M \sum_{\substack{m'=1\\m'\neq m}}^M  x_{mi}x_{m'i} q_{mk} q_{mi} q_{m'k} q_{m'i}  \zeta_{mki} \zeta_{m'ki}^*,
\end{split}
\end{equation}

\begin{equation}
\label{eq:stat_dot2}
    \E \left[\lvert \dot \gamma_{kk} \rvert ^2\right] =\sum_{m=1}^M \lvert x_{mk}\rvert^2 q_{mk} \zeta_{mkk}^2
    + \sum_{m=1}^M \sum_{\substack{m'=1\\m'\neq m}}^M  x_{mi}x_{m'i} q_{mk} q_{m'k}  \zeta_{mkk} \zeta_{m'kk}^*,
\end{equation}

\begin{equation}
    \E \left[ \lvert \bar\gamma_{kk}   \rvert^2 \right] =\sum_{m=1}^M \lvert x_{mk}\rvert^2\beta_{mk} q_{mk} \zeta_{mkk},
\end{equation}

$$\E \left[ \lvert \hat\gamma_{kk}\rvert^2 \right] = \frac{\left(\E \left[ \lvert \bar\gamma_{kk}\rvert^2 \right] \right)^2} {\E \left[ \lvert \bar\gamma_{kk}\rvert^2 \right] + \sigma^2_d} ,$$
and $$\E \left[ \lvert \tilde\gamma_{kk}\rvert^2 \right] = \frac{\E \left[ \lvert \bar\gamma_{kk}\rvert^2 \right] \sigma^2_d } {\E \left[ \lvert \bar\gamma_{kk}\rvert^2 \right] + \sigma^2_d},$$ 

 with $q_{mk} \triangleq \Pr\{\alpha_{mk}=1\}$, and $\zeta_{mki} \triangleq \mathbf{\dot h}_{mk} ^T \mathbf{\dot h}_{mi}^*$.

\end{theorem}

\begin{proof}

Treating the interference in~\eqref{eq:stat2_recd_sig} as noise, and under the assumption that $\E \left[s_i s_j^*\right] = \delta\left[i-j\right]$,we can write,

\begin{equation}
\label{eq:Rk_stat_2}
    R_k = \E \left[ \log\left(1+ \frac{\lvert \dot\gamma_{kk} + \hat\gamma_{kk}\rvert^2}{  \E \left[ \lvert \tilde\gamma_{kk}   \rvert^2 \right] + \sum^K_{ \substack{i=1 \\ i\neq k}} \E \left[ \lvert\gamma_{ki} \rvert ^2 \right] + \sigma^2_o  }\right)\right].
\end{equation}
Invoking Jensen's inequality, we can upper bound~\eqref{eq:Rk_stat_2} as,

\begin{equation}
    R_k \leq  \log\left( 1+ \frac{\E \left[\lvert \dot\gamma_{kk}\rvert^2\right] + \E \left[\lvert \hat\gamma_{kk}\rvert^2\right]}{ \E \left[ \lvert \tilde\gamma_{kk}   \rvert^2 \right] +\sum^K_{ \substack{i=1 \\ i\neq k}} \E \left[\lvert\gamma_{ki} \rvert ^2\right] + \sigma^2_o  }\right).
\end{equation}

The statistics for the effective direct and interfering downlink channels have been derived in Appendix~\ref{app:stat_stat}. 
\end{proof}

We next validate our derived results via numerical simulations. 
\graphicspath{ {./images/} }
\section{Numerical Results}
\label{sec:results}

In this section we present simulation and numerical results to validate the derived theoretical bounds, and to analyze the downlink performance of the CF-mMIMO system
under LoS/ NLoS channels. Unless stated otherwise, the simulation parameters used for these experiments are listed in Table~\ref{tab:sims}.
\begin{table}[t!]	

	\caption{Simulation Parameters}
	\centering
	\begin{tabular}{|c|c|c|c|}
		\hline 
		\bf{Parameter} & \bf{Value} & \bf{Parameter} & \bf{Value}\\
		\hline 
		Reference distance ($d_0$) & 1~\si{\meter} & Network Area & $1$ \si{\square\kilo\meter}=$10^6$ \si{\square\meter}
		\\ Number of Users~($K$) & 64 & Number of APs (M) & 1024
		\\ UE antenna height & 1.5~\si{\meter}& AP Height & 10~\si{\meter}
		\\ Carrier frequency ($f_c$) & 3.5 GHz & $\eta$& 0.5 
		\\ $\mu$& $300/$\si{\square\kilo\meter} & Average building height ($\rho$) & 20~\si{\meter}
		\\$N$ & 1 & Normalised Signal Noise Power ($\sigma_o$) & 20dB
	 \\ \hline 
			\end{tabular}
		\label{tab:sims}
\end{table} 
\begin{figure}[t!]
	\centering
	\includegraphics[width=0.7\textwidth]{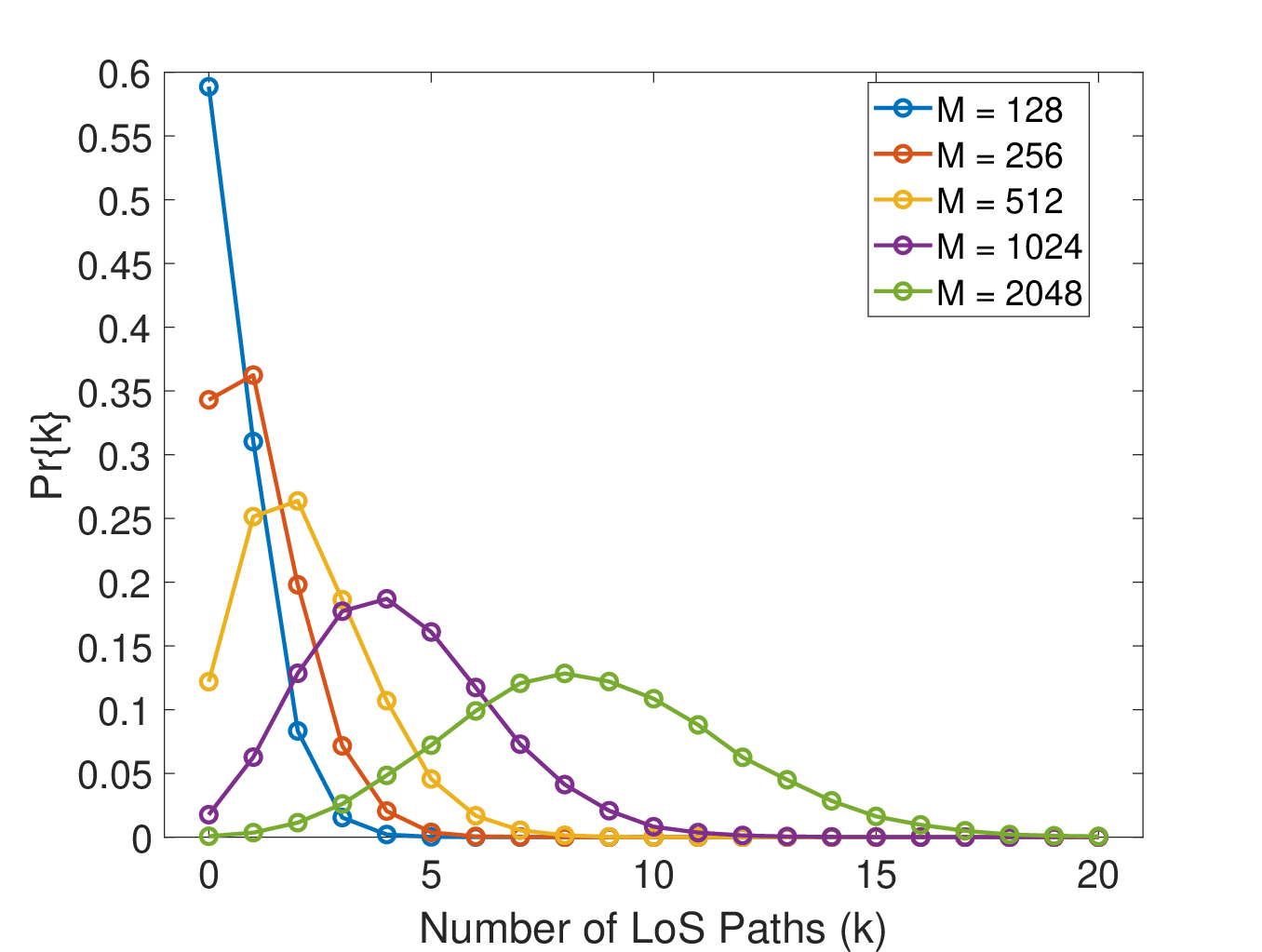}
	\caption{PMF of number of LoS channels per UE for varying AP densities}
	\label{fig:p}
\end{figure}
 
In Fig.~\ref{fig:p}, we plot the probability mass function~(PMF) for
the number of LoS channels for an arbitrarily placed UE. We observe
that increasing the number of APs from 128 to 1024 dramatically
reduces the probability of the absence of an LoS channel from ~60\% to~2\%. This, coupled with the deterministic and high gain nature of LoS
channels motivates us to deploy more single antenna APs instead of a
smaller number of APs with multiple antennas. Therefore, for AP
densities greater than $10^3$ APs
/\si{\square\kilo\meter}, transmission schemes that rely heavily on LoS communication
can be deployed effectively without significantly affecting the 5th
percentile performance.
 
 \begin{figure}[t]
  \centering
    \includegraphics[width=0.7\textwidth]{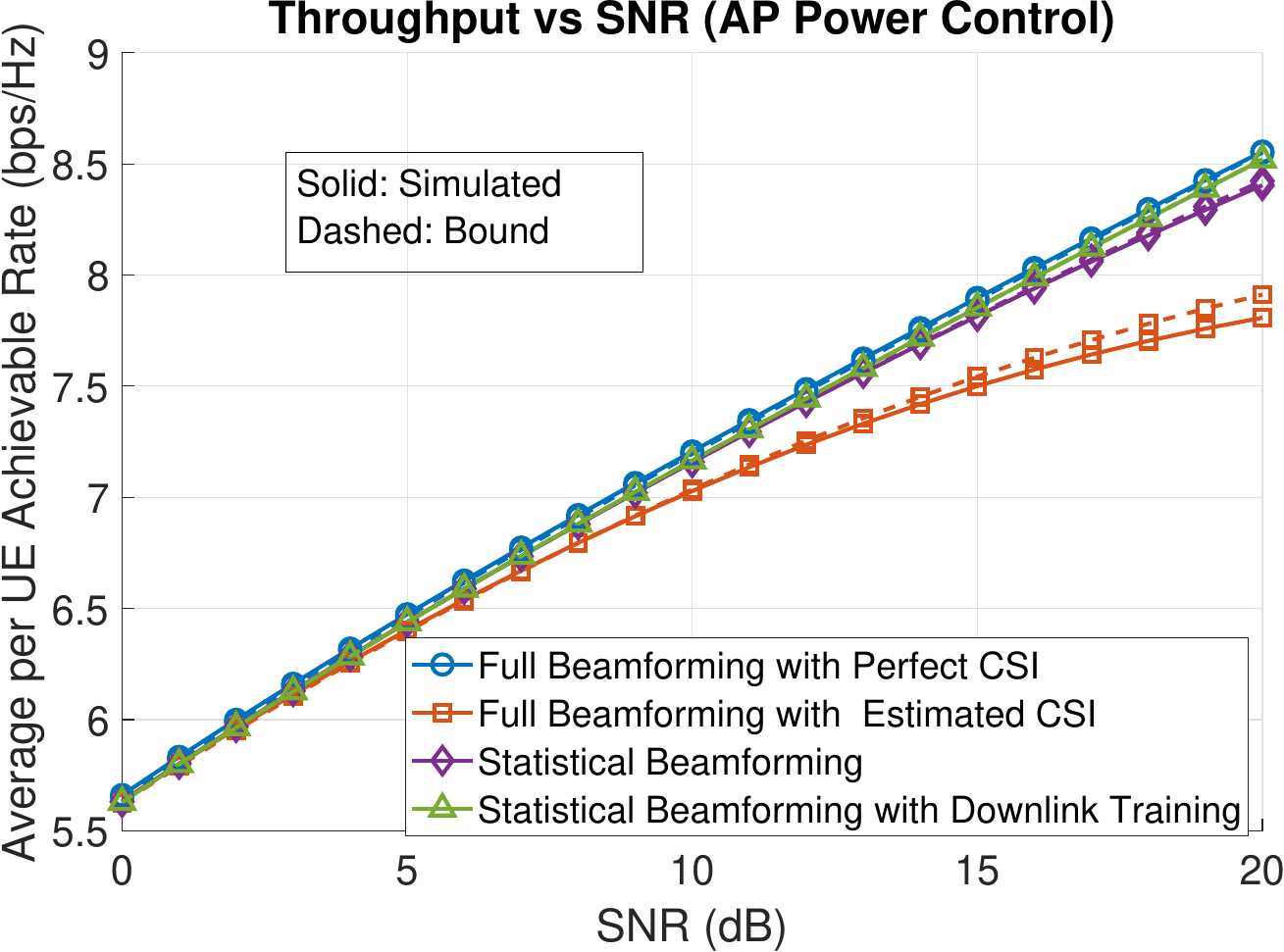}
     \caption{Average per user achievable rate as a function of the data SNR for different beamforming schemes with a fixed amount of power being transmitted by each AP.}
    \label{fig:t1a}
 \end{figure}   
  \begin{figure}
  	\centering
    \includegraphics[width=0.7\textwidth]{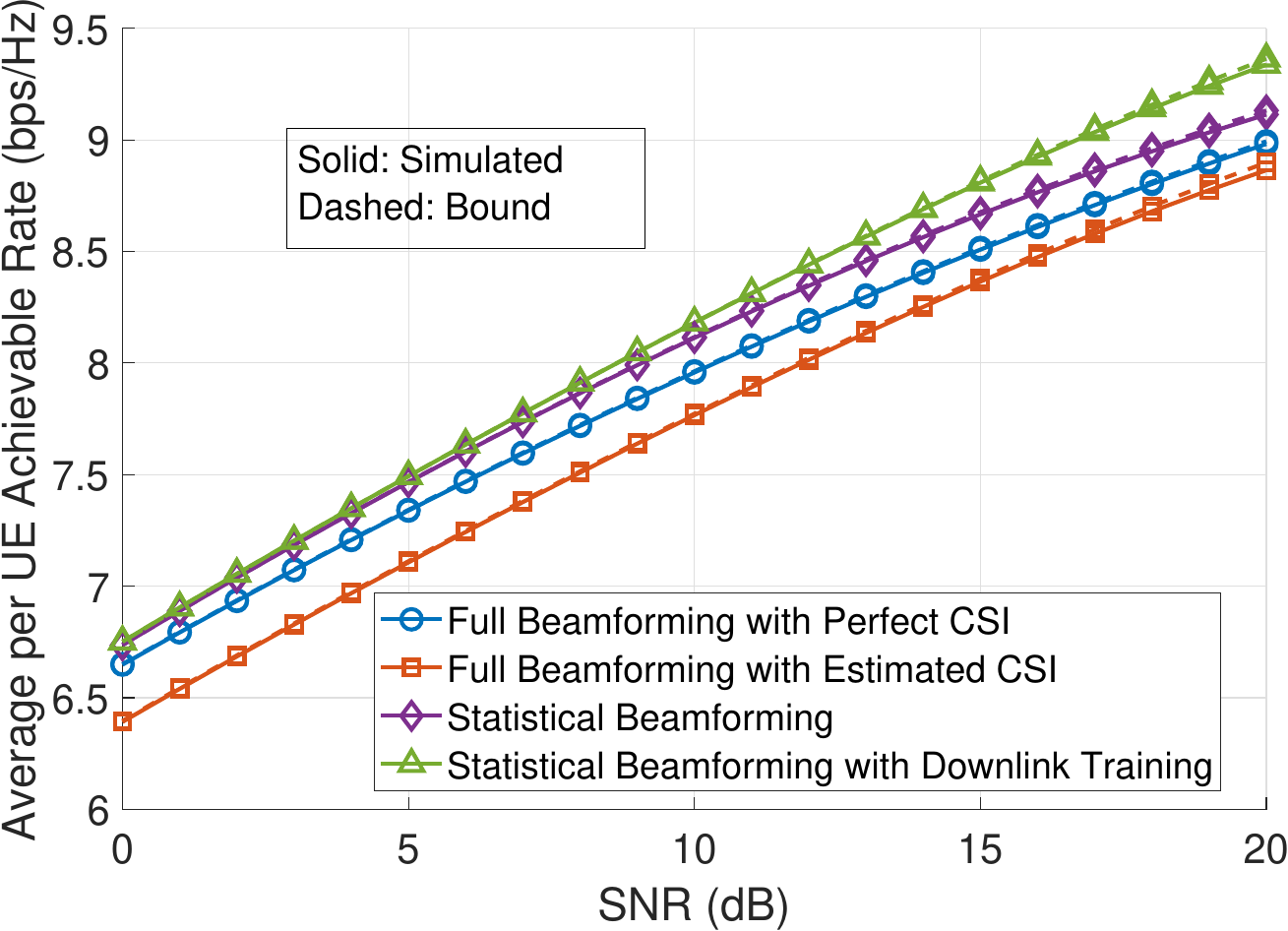}
    \caption{Average per user achievable rate as a function of the data SNR for different beamforming schemes with a fixed amount of power being transmitted to each UE.}
    \label{fig:t1b}
 \end{figure}
 
In Fig.~\ref{fig:t1a}, we plot the average achievable rate per user
against the normalized transmitted data power  while constraining the
power transmitted by each AP. We note that the performance of
statistical beamforming with downlink training matches closely with
the performance achieved under the availability of accurate CSI. This
is due to the fact that, under statistical beamforming, the APs
allocate no power to the NLoS channels between themseleves and the
UEs, reducing the inter-stream interference, and improving the
SINR. It is important to note that, since there is a 76\% probability
that a certain AP has no LoS channels to UEs, it will remain silent,
thus making statistical beamforming more power efficient. We keep the maximum power of the system constant across power allocation schemes, and the respective maximum power constants are set as, $\frac{\mathcal{E}_{u,max}}{\mathcal{E}_{a,max}}= \frac{M}{K}$.

The average user throughput as a function of the normalized
transmitted data SNR while constraining the power transmitted to each
UE is plotted in Fig.~\ref{fig:t1b}. We observe that, in this case,
statistical beamforming, both with and without downlink training,
outperforms full beamforming. This is due to the fact that in the
former case, the power is allocated only to the LoS channels rather
than all the possible channels to a user, most of which are NLoS. We
also observe that all the schemes other than statistical beamforming
with downlink training seem to saturate at high SNRs. In case of full
beamforming, this is explained by high levels of inter-stream
interference due to all $1024 \times 64$ channels being in use. On the other
hand, pure statistical beamforming suffers from strong
self-interference since the scheme neglects the fast-fading channel
components. 

\begin{figure}
  \centering
    \includegraphics[width=0.7\textwidth]{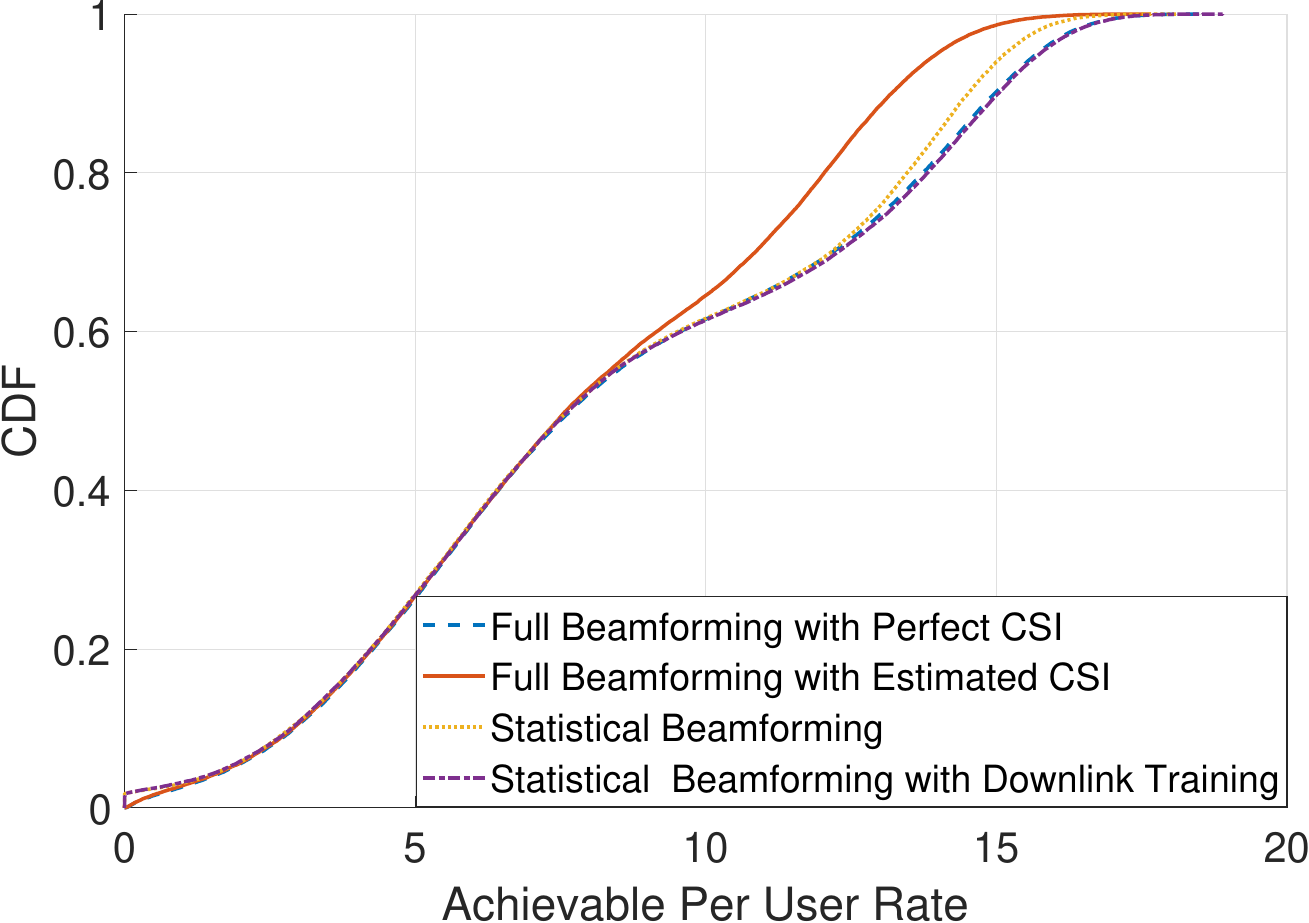}
    \caption{Simulated CDF of the per user spectral efficiency for different beamforming schemes and a fixed amount of power being transmitted by each AP.}
    \label{fig:c1a}
  \end{figure}
  \begin{figure}
  	\centering
    \includegraphics[width=0.7\textwidth]{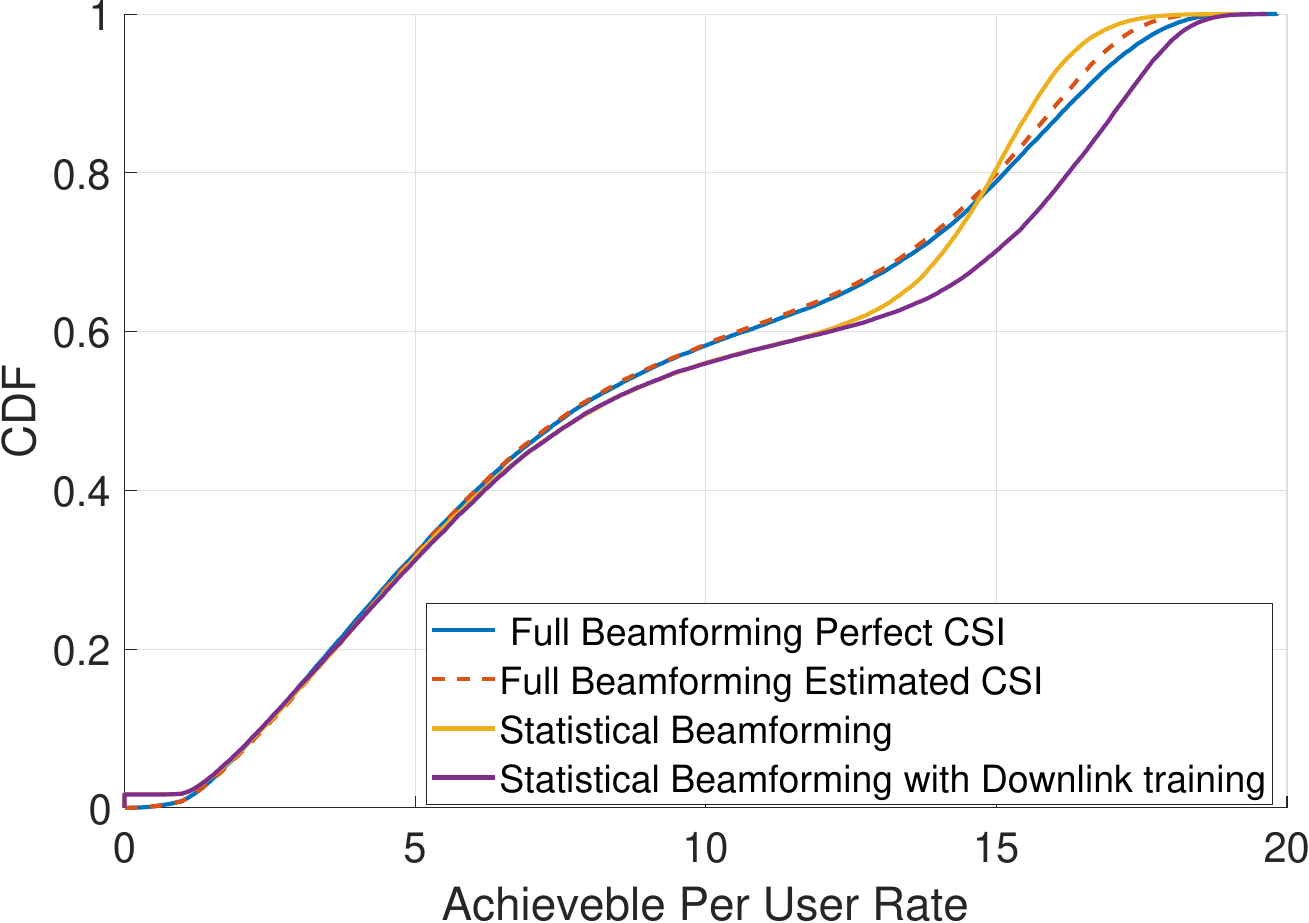}
    \caption{Simulated CDF of the per user spectral efficiency for different beamforming schemes and a fixed amount of power being transmitted to each UE.}
    \label{fig:c1b}
\end{figure}

In addition to the average throughputs, it is also necessary to
analyse the distribution functions of the user throughputs for all the
schemes, as done in in Figs.~\ref{fig:c1a} and~\ref{fig:c1b}. We
observe that, as expected, the statistical beamforming schemes have an
abrupt jump of about 2\% at rate 0. This corresponds to the fraction of
users that do not have an available LoS channel as observed in Fig.~\ref{fig:p}. It interesting to note that for both the power
control schemes the performance of users from the 5th percentile to
about 50th percentile is identical.

In the case where the power being transmitted by each AP remains
fixed, the performance of the high throughput users is commensurate
with the observations from the achievable rate vs. SNR plots
(Figs.~\ref{fig:t1a} and~\ref{fig:t1b}). It is interesting to note how
closely the CDFs of the statistical schemes track the accurate CSI
case. Apart from the fairly low performance loss at high throughput,
the only major limitation is the abrupt 2\% jump observed at 0 due to the
unavoidable outage. This outage is because of the scheme's reliance on high AP
densities.

In case we limit the power transmitted to each UE, the statistical
schemes work at full power capacity, unlike in the case where we limit
the power transmitted by each AP. Constraining the power transmitted to each UE allows the CPU to divert the bulk of its power to APs located near active users. This fact, coupled with the reliance of statistical schemes on LoS transmissions makes it a superior choice for high throughput users. We also note that the self interference experienced by the pure statistical transmission at high SNRs make this scheme interference limited and saturate its performance. 
\begin{figure}
  \centering
    \includegraphics[width=0.7\textwidth]{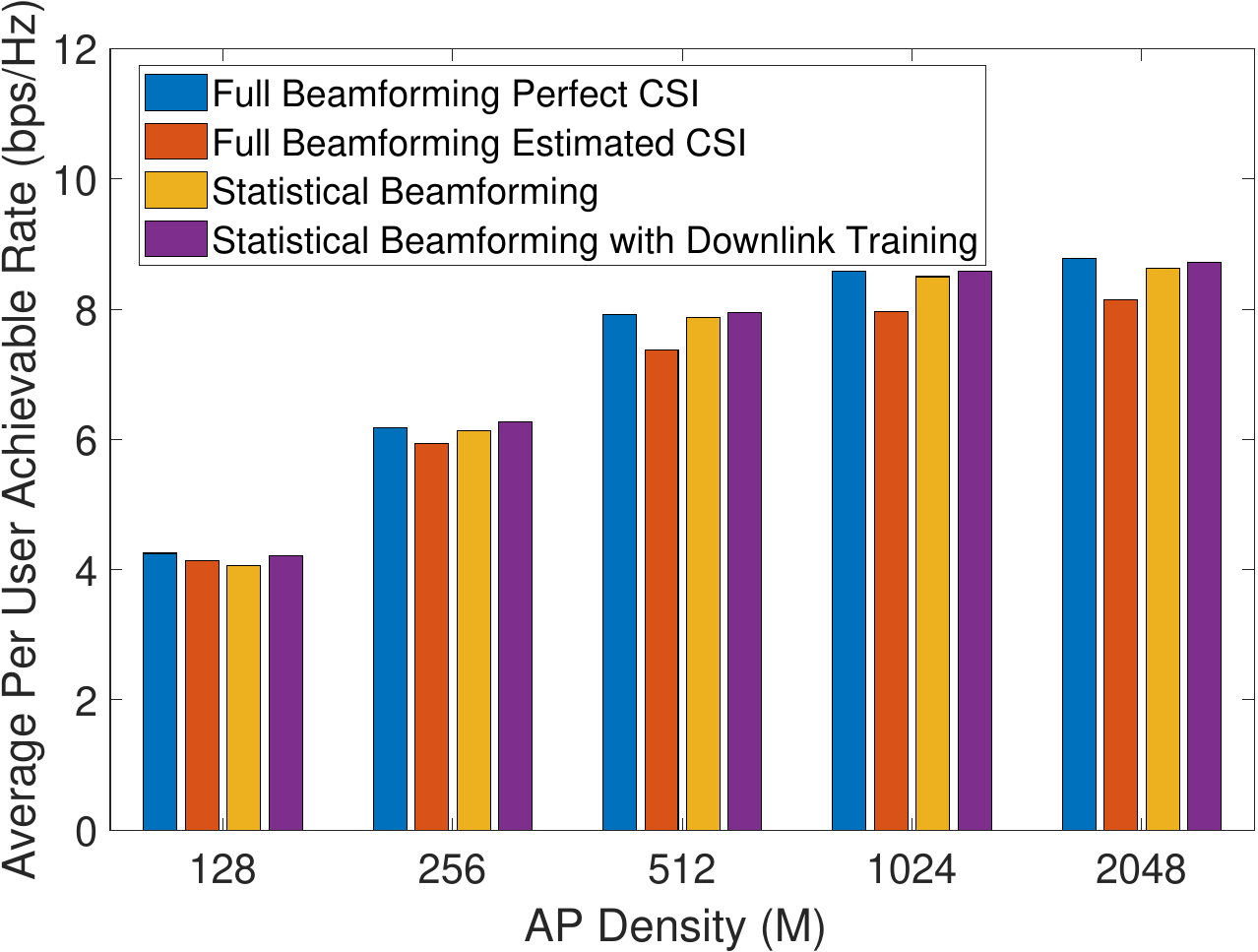}
    \caption{Average per user achievable rate as a function of the AP density for different beamforming schemes with a fixed amount of power being transmitted by each AP.}
    \label{fig:b1a}
  \end{figure}
  \begin{figure}[t!]
 \centering    
\includegraphics[width=0.7\textwidth]{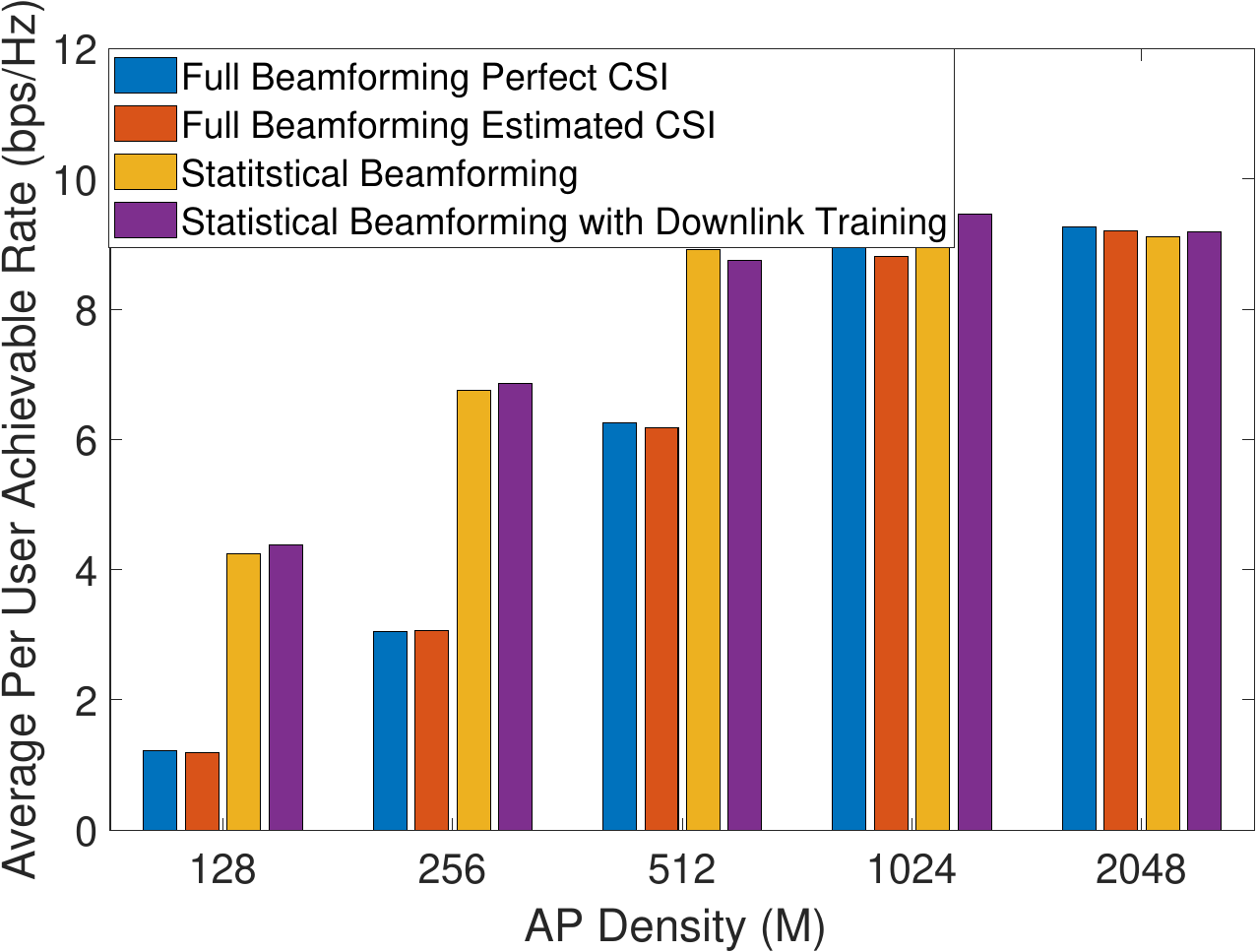}
    \caption{Average per user achievable rate as a function of the AP density for different beamforming schemes with a fixed amount of power being transmitted to each UE.}
    \label{fig:b1b}
 \end{figure}

 In Figs.~\ref{fig:b1a} and~\ref{fig:b1b} we plot the average
 throughput of each precoding scheme with AP density varying from 128
 AP/\si{\square\kilo\meter} to 2048 AP/\si{\square\kilo\meter} for both the power control schemes. At
 first, it is apparent that the performance gains from increasing the
 AP density start saturating at about $10^3$~AP/\si{\square\kilo\meter}. This can be attributed to the increased
 inter-stream interference. In Fig.~\ref{fig:b1a}, where the case of
 constraining the power transmitted by each AP is considered, the
 observations are mostly in line with Fig.~\ref{fig:t1a}. Here,
 accurate CSI precoding performs the best, followed by statistical
 beamforming with downlink training, pure statistical beamforming and
 finally, estimated CSI precoding. The only exception is at $M=128$,
 where pure statistical beamforming performs worse than estimated CSI
 precoding. This is to be expected, since only about 40\% of UEs have
 an LoS channel. We again note that statistical beamforming is
 operating at only 24\% of the total power available, as opposed to
 the full beamforming schemes that use all the available
 power. For the case where power transmitted to each UE is constrained
 (refer Fig.~\ref{fig:b1b}), since the statistical schemes utilize
 nearly all the power available, their performance far exceeds that
 obtained with full beamforming, especially for low AP
 densities. It must be noted that, since a relatively high proportion
 of UEs have no LoS channels at low AP densities, we expect a
 trade-off in their CDFs. This overwhelming gain in throughput at low
 AP densities, as mentioned earlier, is because all the energy of the
 system is being diverted to APs with LoS channels near active
 users. However, this phenomenon causes a reduction in the throughput
 of the statistical beamforming schemes at an AP density of
 $M=2048$.

It must be noted that the statistical beamforming scheme without downlink training has the unique advantage of not allocating any resources for uplink/ downlink training, and using the entire frame for data transmission. 

\section{Conclusion}
In this paper, we have discussed the downlink performance of CF-mMIMO
systems under probabilistic LoS/ NLoS channels. We have evaluated the
performance of these systems under full, as well as statistical
beamforming schemes. At sufficiently high AP densities, it was
observed that statistical beamforming schemes provide superior
throughput with little or no loss in coverage and at much higher
energy efficiency. Furthermore, for an AP density of 1024
AP/\si{\square\kilo\meter} the achievable rates of users between the
fifth and the fiftieth percentile is nearly identical, with only 2\% of the
users in operating in the statistical schemes experiencing zero
throughput, and most of the variation being among users who are among the top half in terms of achievable rates. We also show that the
performance and coverage gains obtained from increasing the AP density
saturates at about 1000 AP/\si{\square\kilo\meter}. 
This shows that LoS reliant statistical beamforming schemes can lead
to performance that is commensurate or even superior to full
beamforming. 
Further research would focus on optimal power control schemes for
CF-mMIMO systems under probabilistic LoS/NLoS channel conditions as
well as on precoding schemes that are intended to mitigate
inter-stream interference such as regularised zero
forcing.


\begin{appendix}
\subsection{Channel Statistics with accurate CSI}
\label{app:stat_perfect}
We know that, $\mathbf{ h}_{mk} = \alpha_{mk} \mathbf{\dot h}_{mk} + \sqrt{\beta_{mk}} \mathbf{\bar h}_{mk}$, where $\alpha_{mk}$ is a Bernoulli RV with $\Pr\{\alpha_{mk}=1\}=q_{mk}$. Defining $\zeta_{mki} = \mathbf{\dot h}_{mk} ^T \mathbf{\dot h}_{mi}^*$, we can write,






\begin{equation}
    \E \left[\gamma_{ki}\right] = \sum_{m=1} ^M x_{mi} q_{mk} q_{mi} \zeta_{mki} + x_{mk} N \beta_{mk} \delta[k-i].
\end{equation}
Similarly, 
\begin{equation}
\E \left[\lvert\gamma_{kk}\rvert ^2\right] = \sum _{m=1} ^M |x_{mk}|^2 \E \left[|\mathbf{h}_{mk}^T \mathbf{h}_{mk}^*|^2\right] 
\\+\sum_{m=1}^M \sum_{\substack{m'=1\\m'\neq m}}^M x_{mk}x_{m'k}  \E \left[ \mathbf{h}_{mk}^T \mathbf{h}_{mk}^*\right]\E \left[ \mathbf{h}_{m'k}^T \mathbf{h}_{m'k}^*\right]^*,
\end{equation}
where $\E \left[ \mathbf{h}_{mk}^T \mathbf{h}_{mk}^*\right] = q_{mk} \zeta_{mkk} + \beta_{mk} N  $ and,
\begin{align*}
   \E \left[|\mathbf{h}_{mk}^T \mathbf{h}_{mk}^*|^2\right] &=
   \E \left[ \lvert \alpha_{mk} \mathbf{\dot h}_{mk}^T  \mathbf{\dot h}_{mk}^* + \beta_{mk} \mathbf{\bar h}_{mk}^T  \mathbf{\bar h}_{mk}^* + \alpha_{mk}  \sqrt{\beta_{mk}} (\mathbf{\bar h}_{mk}^T  \mathbf{\dot h}_{mk}^* + \mathbf{\dot h}_{mk}^T  \mathbf{\bar h}_{mk}^*)  \rvert ^2 \right]
   \\
  &= \E \left[ \alpha_{mk} \lvert  \mathbf{\dot h}_{mk}^T \mathbf{\dot h}_{mk}^* \rvert^2 + \beta_{mk}^2  \lvert  \mathbf{\bar h}_{mk}^T \mathbf{\bar h}_{mk}^*  \rvert^2  + 2 \alpha_{mk} \beta_{mk} \lvert \mathbf{\dot h}_{mk}^T \mathbf{\bar h}_{mk}^* \rvert^2 
  \right. \\ & \left. 
  + 2 \Re (  \alpha_{mk} \beta_{mk} \mathbf{\dot h}_{mk}^T  \mathbf{\dot h}_{mk}^* \mathbf{\bar h}_{mk}^T  \mathbf{\bar h}_{mk}^* + \alpha_{mk} \sqrt{\beta_{mk}} \mathbf{\dot h}_{mk}^T  \mathbf{\dot h}_{mk}^* \mathbf{\dot h}_{mk}^T  \mathbf{\bar h}_{mk}^* 
   \right. \\ & \left.
  + \alpha_{mk} \sqrt{\beta_{mk}} \mathbf{\dot h}_{mk}^T  \mathbf{\dot h}_{mk}^* \mathbf{\dot h}_{mk}^H  \mathbf{\bar h}_{mk} 
  + \alpha_{mk}\beta_{mk}^{\frac{3}{2}} \mathbf{\bar h}_{mk}^T  \mathbf{\bar h}_{mk}^* \mathbf{\bar h}_{mk}^H \mathbf{\dot h}_{mk}
   \right.   \\ & \left.
  + \alpha_{mk}\beta_{mk}^{\frac{3}{2}} \mathbf{\bar h}_{mk}^T  \mathbf{\bar h}_{mk}^* \mathbf{\bar h}_{mk}^T \mathbf{\dot h}_{mk}^* + \alpha_{mk} \beta_{mk} (\mathbf{\bar h}_{mk}^T  \mathbf{\dot h}_{mk}^*)^2   )\right]
  \\ 
  &= q_{mk} \zeta_{mkk} ^2 + N (N+1)\beta_{mk}^2 + 2 (N+1) q_{mk} \beta_{mk} \zeta_{mkk}.
\end{align*}
This simplifies to the expression in~\eqref{eq:stat_acc_kk}.

Similarly for $\gamma_{ki}$ and $k\neq i$,
\begin{equation}
\E \left[\lvert\gamma_{ki}\rvert ^2\right] = \sum _{m=1} ^M |x_{mi}|^2 \E \left[|\mathbf{h}_{mk}^T \mathbf{h}_{mi}^*|^2\right] 
+\sum_{m=1}^M \sum_{\substack{m'=1\\m'\neq m}}^M x_{mi}x_{m'i} \E \left[ \mathbf{h}_{mk}^T\right] \E \left[\mathbf{h}_{mi}^*\right]\E \left[ \mathbf{h}_{m'k}^T \right]^* \E \left[\mathbf{h}_{m'i}^*\right]^*,
\end{equation}
where, $\E \left[\mathbf{h}_{mk}\right]= q_{mk} \mathbf{\dot h}_{mk}$ and,
\begin{align*}
\E \left[|\mathbf{h}_{mk}^T \mathbf{h}_{mi}^*|^2\right] &= \E \left[ \lvert \alpha_{mk}\alpha_{mi} \mathbf{\dot h}_{mk}^T \mathbf{\dot h}_{mi}^*  + \sqrt{\beta_{mk} \beta_{mi}} \mathbf{\bar h}_{mk}^T \mathbf{\bar h}_{mi}^* 
+ \alpha_{mk}\sqrt{\beta_{mi}} \mathbf{\dot h}_{mk}^T \mathbf{\bar h}_{mi}^* \right. \\& \left.
+ \sqrt{\beta_{mk}} \alpha_{mi} \mathbf{\bar h}_{mk}^T \mathbf{\dot h}_{mi}^* \rvert ^2\right] 
\\
&= \E \left[  \alpha_{mk}\alpha_{mi} \lvert \mathbf{\dot h}_{mk}^T \mathbf{\dot h}_{mi}^*\rvert^2  + \beta_{mk} \beta_{mi} \lvert \mathbf{\bar h}_{mk}^T \mathbf{\bar h}_{mi}^*\rvert^2 + \alpha_{mk}\beta_{mi} \lvert\mathbf{\dot h}_{mk}^T \mathbf{\bar h}_{mi}^* \rvert^2
\right.\\ &\left.
+ \beta_{mk} \alpha_{mi} \lvert \mathbf{\bar h}_{mk}^T \mathbf{\dot h}_{mi}^* \rvert ^2 
+2 \Re \Big( \alpha_{mk}\alpha_{mi} \sqrt{\beta_{mk} \beta_{mi}} \mathbf{\dot h}_{mk}^T \mathbf{\dot h}_{mi}^* \mathbf{\bar h}_{mk}^H \mathbf{\bar h}_{mi} 
\right.\\ &\left.
+ \alpha_{mk}\alpha_{mi} \sqrt{\beta_{mi}} \mathbf{\dot h}_{mk}^T \mathbf{\dot h}_{mi}^* \mathbf{\dot h}_{mk}^H \mathbf{\bar h}_{mi} 
+ \alpha_{mk}\alpha_{mi} \sqrt{\beta_{mk}} \mathbf{\dot h}_{mk}^T \mathbf{\dot h}_{mi}^* \mathbf{\dot h}_{mi}^T \mathbf{\bar h}_{mk}^*
\right.\\ &\left. 
+\alpha_{mk} \sqrt{\beta_{mk}} \beta_{mi}  \mathbf{\bar h}_{mk}^T \mathbf{\bar h}_{mi}^* \mathbf{\bar h}_{mi}^T \mathbf{\dot h}_{mk}^* 
+  \alpha_{mk} \beta_{mk} \sqrt{\beta_{mi}}  \mathbf{\bar h}_{mk}^T \mathbf{\bar h}_{mi}^* \mathbf{\bar h}_{mk}^H \mathbf{\dot h}_{mi} 
\right.\\ &\left. 
+ \alpha_{mk}\alpha_{mi} \sqrt{\beta_{mk} \beta_{mi}} \mathbf{\dot h}_{mk}^T  \mathbf{\bar h}_{mi}^* \mathbf{\bar h}_{mk}^H \mathbf{\dot h}_{mi} \Big) \right]
\\
&= q_{mk}q_{mi}  \lvert \zeta_{mki} \rvert ^2  + \beta_{mk} \beta_{mi} N + q_{mk} \beta_{mi} \zeta_{mkk}
+ \beta_{mk} q_{mi} \zeta_{mii} .
\end{align*}
This simplifies to the expression in~\eqref{eq:stat_acc_ki}.
\subsection{Channel Statistics with Estimated CSI at APs and UEs}
\label{app:stat_est}
\begin{align}
    \E \left[\gamma_{kk}\right] =& \sum_{m=1} ^M x_{mi} \Big( \E \left[ \alpha_{mk} \right] \mathbf{\dot h}_{mk}^T \mathbf{\dot h}_{mk}^* + \sqrt{\beta_{mk}} \E \left[ \alpha_{mk}\right] \E\left[  \mathbf{\bar h}_{mk}^T \right] \mathbf{\dot h}_{mk}^* + \E[ \alpha_{mk}] \sqrt{\beta_{mk}} \mathbf{\dot h}_{mk}^T \E [\mathbf{\hat h}_{mk}^*] \\& \nonumber
+ \beta_{mk}\E [  \mathbf{\bar h}_{mk}^T\mathbf{\hat h}_{mk}^*] \Big) = \sum_{m=1} ^M x_{mk} \Big( q_{mk} \zeta_{mkk} + \frac{N \beta_{mk}}{\beta_{mk}+ \sigma_u^2}\Big)
\end{align}
\begin{align}
    \E \left[\gamma_{ki}\right]  =& \sum_{m=1} ^M x_{mk} \Big( \E \left[ \alpha_{mk} \alpha_{mi} \right] \mathbf{\dot h}_{mk}^T \mathbf{\dot h}_{mi}^* + \sqrt{\beta_{mk}} \E \left[ \alpha_{mi}\right] \E\left[  \mathbf{\bar h}_{mk}^T \right] \mathbf{\dot h}_{mi}^* + \E[\alpha_{mk}]\sqrt{\beta_{mi}} \mathbf{\dot h}_{mk}^T \E [\mathbf{\hat h}_{mi}^*] \\& \nonumber
+ \sqrt{\beta_{mk} \beta}_{mi}\E [  \mathbf{\bar h}_{mk}^T\mathbf{\hat h}_{mi}^*] \Big) =\sum_{m=1} ^M x_{mi} q_{mk} q_{mi}\zeta_{mki}
\end{align}
\begin{equation}
    \E \left[\dot\gamma_{kk}\right] = \sum_{m=1} ^M x_{mk} \E \left[ \alpha_{mk} \right] \mathbf{\dot h}_{mk}^T \mathbf{\dot h}_{mk}^* =\sum_{m=1} ^M x_{mk} q_{mk} \zeta_{mkk}
\end{equation}
\begin{equation}
     \E \left[\bar\gamma_{kk}\right] =  \E \left[\gamma_{kk}\right] -  \E \left[\dot\gamma_{kk}\right]  = \sum_{m=1} ^M x_{mk} N \frac{ \beta_{mk}}{\beta_{mk}+ \sigma_u^2}
\end{equation}

Similarly, 
\begin{align}
\E \left[\lvert\gamma_{kk}\rvert ^2\right] = \sum _{m=1} ^M |x_{mk}|^2 \E \left[|\mathbf{h}_{mk}^T \mathbf{\check h}^{*}_{mk}|^2\right] 
+\sum_{m=1}^M \sum_{\substack{m'=1\\m'\neq m}}^M x_{mk}x_{m'k}  \E \left[ \mathbf{h}_{mk}^T \mathbf{\check h}^{*}_{mk}\right]\E \left[ \mathbf{h}_{m'k}^T \mathbf{\check h}^{*}_{m'k}\right]^*,
\end{align}

where $\E \left[ \mathbf{h}_{mk}^T \mathbf{\check h}^{*}_{mk}\right] = q_{mk}\zeta_{mkk} + N\beta_{mk}\frac{ \beta_{mk}}{\beta_{mk}+ \sigma_u^2}   $ and,
\begin{align*}
   \E \left[|\mathbf{h}_{mk}^T \mathbf{\check h}^{*}_{mk}|^2\right] &= \E \left[ \lvert \alpha_{mk} \mathbf{\dot h}^T_{mk} \mathbf{\dot h}^*_{mk} + \beta_{mk} \mathbf{\bar h}^T_{mk} \mathbf{\hat h}^*_{mk} + \alpha_{mk} \sqrt{\beta_{mk}} (\mathbf{\dot h}^T_{mk} \mathbf{\hat h}^*_{mk} + \mathbf{\bar h}^T_{mk} \mathbf{\dot h}^*_{mk}) \rvert ^2 \right]
  \\&
   = \E \left[ \alpha_{mk} \lvert  \mathbf{\dot h}_{mk}^T \mathbf{\dot h}_{mk}^* \rvert^2 + \beta_{mk}^2  \lvert  \mathbf{\bar h}_{mk}^T \mathbf{\hat h}_{mk}^*  \rvert^2  + \alpha_{mk} \beta_{mk} \lvert \mathbf{\dot h}_{mk}^T \mathbf{\hat h}_{mk}^* \rvert^2 
   \right. \\& \left. 
   + \alpha_{mk} \beta_{mk} \lvert \mathbf{\bar h}_{mk}^T \mathbf{\dot h}_{mk}^* \rvert^2 
   + 2 \Re \Big(  \alpha_{mk} \beta_{mk} \mathbf{\dot h}_{mk}^T  \mathbf{\dot h}_{mk}^* \mathbf{\hat h}_{mk}^T  \mathbf{\bar h}_{mk}^* 
   \right. \\& \left. 
  + \alpha_{mk} \sqrt{\beta_{mk}} \mathbf{\dot h}_{mk}^T  \mathbf{\dot h}_{mk}^* \mathbf{\dot h}_{mk}^H  \mathbf{\hat h}_{mk} + \alpha_{mk} \sqrt{\beta_{mk}} \mathbf{\dot h}_{mk}^T  \mathbf{\dot h}_{mk}^* \mathbf{\dot h}_{mk}^T  \mathbf{\bar h}_{mk}^* 
  \right. \\& \left. 
  +\alpha_{mk}\beta_{mk}^{\frac{3}{2}} \mathbf{\bar h}_{mk}^T  \mathbf{\hat h}_{mk}^* \mathbf{\dot h}_{mk}^H \mathbf{\hat h}_{mk}
  + \alpha_{mk}\beta_{mk}^{\frac{3}{2}} \mathbf{\bar h}_{mk}^T  \mathbf{\hat h}_{mk}^* \mathbf{\bar h}_{mk}^H \mathbf{\dot h}_{mk} 
  \right. \\& \left. 
  + \alpha_{mk} \beta_{mk} \mathbf{\dot h}_{mk}^T  \mathbf{\hat h}_{mk}^*  \mathbf{\dot h}_{mk}^T  \mathbf{\bar h}_{mk}^* \Big)\right]
  \\& = q_{mk} \zeta_{mkk}^2 + N(N+1)\beta_{mk}^2 \left(\frac{\beta_{mk}}{\beta_{mk}+\sigma_u^2}\right)^2
  \\&
  +N \beta_{mk}^2 \left(\frac{\sqrt{\beta_{mk}} \sigma_u}{\beta_{mk}+\sigma_u^2}\right)^2 
  + q_{mk} \beta_{mk} \Bigg( 1 + \frac{\beta_{mk}}{\beta_{mk}+\sigma_u^2}\Bigg) \zeta_{mkk} 
  \\&
  + 2 N q_{mk}  \beta_{mk} \frac{\beta_{mk} }{\beta_{mk} + \sigma_u^2}  \zeta_{mkk}.
\end{align*}
Therefore,
\begin{align}
\E \left[\lvert\gamma_{kk}\rvert ^2\right] &= \sum _{m=1} ^M |x_{mi}|^2 \Bigg(q_{mk} \zeta_{mkk}^2 + N(N+1)\beta_{mk}^2 \left(\frac{\beta_{mk}}{\beta_{mk}+\sigma_u^2}\right)^2  
  +N \beta_{mk}^2 \left(\frac{\sqrt{\beta_{mk}} \sigma_u}{\beta_{mk}+\sigma_u^2}\right)^2 
  \\& \nonumber
  + q_{mk} \beta_{mk} ( 1 + \frac{\beta_{mk}}{\beta_{mk}+\sigma_u^2}) \zeta_{mkk} 
  + 2 N q_{mk}  \beta_{mk} \frac{\beta_{mk} }{\beta_{mk} + \sigma_u^2}  \zeta_{mkk} \Bigg)
  \\& \nonumber
+\sum_{m=1}^M \sum_{\substack{m'=1\\m'\neq m}}^M x_{mi}x_{m'i} \left( q_{mk}\zeta_{mkk} + N\beta_{mk}\frac{ \beta_{mk}}{\beta_{mk}+ \sigma_u^2} \right)
  \\& \nonumber
\times\left( q_{m'k}\zeta_{m'kk} + N\beta_{m'k}\frac{ \beta_{m'k}}{\beta_{m'k}+ \sigma_u^2} \right)^*.
\end{align}
Similarly for $\E \left[\lvert \gamma_{ki}\rvert^2\right]$,
\begin{multline}
\E \left[\lvert\gamma_{ki}\rvert ^2\right] = \sum _{m=1} ^M |x_{mi}|^2 \E \left[|\mathbf{h}_{mk}^T \mathbf{\check h}^{*}_{mi}|^2\right] 
+\sum_{m=1}^M \sum_{\substack{m'=1\\m'\neq m}}^M x_{mi}x_{m'i}  \E \left[ \mathbf{h}_{mk}^T \right]\E \left[ \mathbf{\check h}^{*}_{mi}\right]\E \left[ \mathbf{h}_{m'k}^T \right]^*\E \left[\mathbf{\check h}_{m'i}^*\right]^*,
\end{multline}
where, $\E \left[\mathbf{h}_{mk}\right] =\E \left[\mathbf{\check h}_{mk}\right] = \E \left[\alpha_{mk} \mathbf{\dot h}_{mk}\right] = q_{mk} \mathbf{\dot h}_{mk} $ and,

\begin{align*}
\E \left[|\mathbf{h}_{mk}^T \mathbf{\check h}^{*}_{mi} |^2\right] &= \E \left[ \lvert \alpha_{mk} \alpha_{mi} \mathbf{\dot h}^T_{mk} \mathbf{\dot h}^*_{mi} + \sqrt{\beta_{mk}\beta_{mi}} \mathbf{\bar h}^T_{mk} \mathbf{\hat h}^*_{mi}
+ \alpha_{mk} \sqrt{\beta_{mi}} \mathbf{\dot h}^T_{mk} \mathbf{\hat h}^*_{mi} 
\right. \\& \left.
+\alpha_{mi} \sqrt{\beta_{mk}} \mathbf{\bar h}^T_{mk} \mathbf{\dot h}^*_{mi} \rvert ^2 \right]
\\&
= \E \left[ \alpha_{mk} \alpha_{mi} \lvert  \mathbf{\dot h}_{mk}^T \mathbf{\dot h}_{mi}^* \rvert^2 + \beta_{mk} \beta_{mi}  \lvert  \mathbf{\bar h}_{mk}^T \mathbf{\hat h}_{mi}^*  \rvert^2  + \alpha_{mk} \beta_{mi} \lvert \mathbf{\dot h}_{mk}^T \mathbf{\hat h}_{mi}^* \rvert^2
\right. \\& \left. 
+ \alpha_{mi} \beta_{mk} \lvert \mathbf{\bar h}_{mk}^T \mathbf{\dot h}_{mi}^* \rvert^2 
   + 2 \Re (  \alpha_{mk} \alpha_{mi} \sqrt{\beta_{mk} \beta_{mi}}\mathbf{\dot h}_{mk}^T  \mathbf{\dot h}_{mi}^* \mathbf{\bar h}_{mk}^H  \mathbf{\hat h}_{mi}
   \right. \\& \left. 
  + \alpha_{mk} \alpha_{mi}\sqrt{\beta_{mi}} \mathbf{\dot h}_{mk}^T  \mathbf{\dot h}_{mi}^* \mathbf{\dot h}_{mk}^H  \mathbf{\hat h}_{mi} 
  + \alpha_{mk} \alpha_{mi} \sqrt{\beta_{mk}} \mathbf{\dot h}_{mk}^T  \mathbf{\dot h}_{mi}^* \mathbf{\bar h}_{mk}^H  \mathbf{\dot h}_{mi} 
  \right. \\& \left. 
  +\alpha_{mk}\sqrt{\beta_{mk}}\beta_{mi} \mathbf{\bar h}_{mk}^T  \mathbf{\hat h}_{mi}^* \mathbf{\dot h}_{mk}^H \mathbf{\hat h}_{mi}
  + \alpha_{mk}\beta_{mk} \sqrt{\beta_{mi}} \mathbf{\bar h}_{mk}^T  \mathbf{\hat h}_{mi}^* \mathbf{\bar h}_{mk}^H \mathbf{\dot h}_{mi} 
  \right. \\& \left. 
  + \alpha_{mk} \alpha_{mi} \sqrt{\beta_{mk} \beta_{mi}} \mathbf{\dot h}_{mk}^T  \mathbf{\hat h}_{mi}^*  \mathbf{\bar h}_{mk}^H \mathbf{\dot h}_{mi} )\right]
\\&
= q_{mk}q_{mi} \lvert\zeta_{mki}\rvert^2  +N \beta_{mk} \beta_{mi} \frac{ \beta_{mi}}{\beta_{mi}+\sigma_u^2} 
+ q_{mk} \beta_{mi} \frac { \beta_{mi} }{\beta_{mi}+\sigma_u^2}\zeta_{mkk}
 \\&  + q_{mi} \beta_{mk} \zeta_{mii}.
\end{align*}
This simplifies to the expression in~\eqref{eq:est_acc_ki}.

We calculate $\E \left[\lvert \dot \gamma_{kk}\rvert^2\right]$,
\begin{align}
    \E \left[\lvert \dot \gamma_{kk} \rvert ^2\right] &=\sum_{m=1}^M \lvert x_{mk}\rvert^2\E \left[\lvert \alpha_{mk} \mathbf{\dot h}_{mk}^T \mathbf{\dot h}_{mk}^* \rvert ^2\right] 
    + \sum_{m=1}^M \sum_{\substack{m'=1\\m'\neq m}}^M   x_{mk} x_{m'k} \E \left[\alpha_{mk} \mathbf{\dot h}_{mk}^T  \mathbf{\dot h}_{mk}^*\right] 
    \\& \nonumber
    \times \E \left[ \alpha_{m'k} \mathbf{\dot h}_{m'k}^T  \mathbf{\dot h}_{m'k}^*\right]^*,
\end{align}
where, $ \E \left[ \alpha_{mk} \mathbf{\dot h}_{mk}^T \mathbf{\dot h}_{mk}^*\right] = q_{mk} \zeta_{mkk} $ and,
\begin{equation*}
\begin{split}
 \E \left[ \lvert \alpha_{mk} \mathbf{\dot h}_{mk}^T \mathbf{\dot h}_{mk}^* \rvert ^2 \right] = \E \left[ \lvert \alpha_{mk} \mathbf{\dot h}_{mk}^T \mathbf{\dot h}_{mk}^* *\rvert ^2\right] = q_{mk} \zeta_{mkk}^2.
 \end{split}
\end{equation*}
This expression simplifies to~\eqref{eq:stat_dot},~\eqref{eq:stat_dot1} and \eqref{eq:stat_dot2}.
Finally for $\bar\gamma_{kk}$, 
\begin{equation*}
   \E \left[ \lvert \bar\gamma_{kk}\bar\gamma_{kk}^* \rvert \right] = \E \left[ \rvert  \sum_{m=1} ^M x_{mk} \left( \alpha_{mk} \left( \mathbf{\hat h}_{mk}^T \mathbf{\dot h}_{mk}^* + \mathbf{\dot h}_{mk}^T \mathbf{\bar h}_{mk}^*   \right) + \mathbf{\hat h}_{mk}^T \mathbf{\bar h}_{mk}^* \right) \rvert ^2 \right],
\end{equation*} 
\begin{equation}
\begin{split}
\E \left[\lvert\bar\gamma_{kk}\rvert ^2\right] =& \sum _{m=1} ^M |x_{mk}|^2 \left( \E \left[|\mathbf{h}_{mk}^T \mathbf{\check h}^{*}_{mk}|^2\right] - \E \left[\lvert \alpha_{mk} \mathbf{\dot h}_{mk}^T \mathbf{\dot h}_{mk}^* \rvert ^2\right]   \right)
+\sum_{m=1}^M \sum_{\substack{m'=1\\m'\neq m}}^M x_{mk}x_{m'k} \\& \E \left[ \mathbf{h}_{mk}^T \mathbf{\check h}^{*}_{mk}\right]\E \left[ \mathbf{h}_{m'k}^H \mathbf{\check h}^{*}_{m'k}\right] - \E \left[\alpha_{mk} \mathbf{\dot h}_{mk}^T  \mathbf{\dot h}_{mk}^*\right] \E \left[ \alpha_{m'k} \mathbf{\dot h}_{m'k}^T  \mathbf{\dot h}_{m'k}^*\right]^*,
\end{split}
\end{equation}
where each expectation term has been previously defined.
This expression simplifies to~\eqref{eq:est_bar}.
It is important to note that the two components, $\dot\gamma_{kk}$ and  $\bar\gamma_{kk}$ are uncorrelated. Therefore,
\begin{equation}
    \E \left[ \lvert \gamma_{kk} \rvert ^2\right] = \E \left[ \lvert \dot\gamma_{kk} \rvert ^2\right] + \E \left[ \lvert \bar\gamma_{kk} \rvert ^2\right].
\end{equation}
\subsection{Channel Statistics for Statistical Beamforming Schemes}
\label{app:stat_stat}
\begin{equation}
\begin{split}
    \E \left[\gamma_{kk}\right] &= \sum_{m=1} ^M x_{mk}( \E \left[ \alpha_{mk} \right] \mathbf{\dot h}_{mk}^T \mathbf{\dot h}_{mk}^* + \E \left[ \alpha_{mk}\right] \E\left[  \mathbf{\bar h}_{mk}^T \right] \mathbf{\dot h}_{mk}^*)  =\sum_{m=1} ^M x_{mk} q_{mk} \zeta_{mkk}
\end{split}
\end{equation}
\begin{equation}
\begin{split}
    \E \left[\gamma_{ki}\right] &= \sum_{m=1} ^M x_{mi}( \E \left[ \alpha_{mk} \alpha_{mi} \right] \mathbf{\dot h}_{mk}^T \mathbf{\dot h}_{mi}^* + \E \left[ \alpha_{mi}\right] \E\left[  \mathbf{\bar h}_{mk}^T \right] \mathbf{\dot h}_{mi}^*) =\sum_{m=1} ^M x_{mi} q_{mk} q_{mi} \zeta_{mki}
\end{split}
\end{equation}
\begin{equation}
\begin{split}
    \E \left[\dot\gamma_{kk}\right] = \sum_{m=1} ^M x_{mk} \E \left[ \alpha_{mk} \right] \mathbf{\dot h}_{mk}^T \mathbf{\dot h}_{mk}^* =\sum_{m=1} ^M x_{mk} q_{mk} \zeta_{mkk}
\end{split}
\end{equation}
\begin{equation}
     \E \left[\bar\gamma_{kk}\right] =  \E \left[\gamma_{kk}\right] -  \E \left[\dot\gamma_{kk}\right] = 0
\end{equation}
Similarly,
\begin{align}
    \E \left[\lvert \gamma_{kk} \rvert ^2\right] &=\sum_{m=1}^M \lvert x_{mk}\rvert^2\E \left[\lvert \alpha_{mk} \mathbf{h}_{mk}^T \mathbf{\dot h}_{mk}^* \rvert ^2\right] 
    + \sum_{m=1}^M \sum_{\substack{m'=1\\m'\neq m}}^M  x_{mk} x_{m'k}  \E \left[\alpha_{mk} \mathbf{h}_{mk}^T  \mathbf{\dot h}_{mk}^*\right] \E \left[ \alpha_{m'k} \mathbf{h}_{m'k}^T  \mathbf{\dot h}_{m'k}^*\right]^*,
\end{align}
where, $ \E \left[ \alpha _{mk}\mathbf{h}_{mk}^T \mathbf{\dot h}_{mk}^*  \right]= q_{mk} \zeta_{mkk} $ and,
\begin{align*}
 \E \left[ \lvert \alpha_{mk} \mathbf{h}_{mk}^T \mathbf{\dot h}_{mk}^* \rvert ^2 \right] &= \E \left[ \lvert \alpha_{mk} \mathbf{\dot h}_{mk}^T \mathbf{\dot h}_{mk}^* + \sqrt{\beta_{mk}} \alpha_{mk} \mathbf{\bar h}_{mk}^T \mathbf{\dot h}_{mk} ^*\rvert ^2\right]
 \\
 &=\E \left[ \alpha_{mk} \lvert \mathbf{\dot h}_{mk}^T \mathbf{\dot h}_{mk}^* \rvert ^2 + \alpha_{mk}\beta_{mk} \lvert \mathbf{\bar h}_{mk}^T \mathbf{\dot h}_{mk} ^*\rvert ^2 + \alpha_{mk} \sqrt{\beta_{mk}} 2\Re( \mathbf{\dot h}_{mk}^T \mathbf{\dot h}_{mk}^* \mathbf{\dot h}_{mk}^T \mathbf{\bar h}_{mk}^*)  \right]
 \\ \nonumber
 &= q_{mk} \zeta_{mkk} ^2 + \beta_{mk} q_{mk} \zeta_{mkk}.
\end{align*}
Therefore,
\begin{align}
    \E \left[\lvert \gamma_{kk} \rvert ^2\right] &=\sum_{m=1}^M \lvert x_{mk}\rvert^2\Big( q_{mk} \zeta_{mkk} ^2 + \beta_{mk} q_{mk} \zeta_{mkk} \Big)
    + \sum_{m=1}^M \sum_{\substack{m'=1\\m'\neq m}}^M  x_{mk} x_{m'k} q_{mk} \zeta_{mkk} q_{m'k} \zeta_{m'kk}.
\end{align}
Similarly, for $\E \left[\lvert \gamma_{ki}\rvert^2\right]$,
\begin{align}
    \E \left[\lvert \gamma_{ki} \rvert ^2\right] &=\sum_{m=1}^M \lvert x_{mi}\rvert^2 \E \left[\lvert\alpha_{mi} \mathbf{h}_{mk}^T \mathbf{\dot h}_{mi}^* \rvert ^2\right] 
    \\  & + \sum_{m=1}^M \sum_{\substack{m'=1\\m'\neq m}}^M  x_{mi} x_{m'i} \E \nonumber \left[\mathbf{h}_{mk}^T\right] \E \left[\alpha_{mi}  \mathbf{\dot h}_{mk}^*\right] \E \left[\mathbf{h}_{m'k}^T\right]^* \E \left[\alpha_{m'k} \mathbf{\dot h}_{m'k}^*\right]^*,
\end{align}
where, $\E \left[\mathbf{h}_{mk}\right] = \E \left[\alpha_{mk} \mathbf{\dot h}_{mk}\right] = q_{mk} \mathbf{\dot h}_{mk} $ and,
\begin{multline*}
 \E \left[ \lvert \alpha_{mi} \mathbf{h}_{mk}^T \mathbf{\dot h}_{mi}^* \rvert ^2 \right] = \E \left[ \lvert \alpha_{mk}\alpha_{mi} \mathbf{\dot h}_{mk}^T \mathbf{\dot h}_{mi}^* + \sqrt{\beta_{mk}} \alpha_{mi} \mathbf{\bar h}_{mk}^T \mathbf{\dot h}_{mi} ^*\rvert ^2\right]
  \\=\E \left[ \alpha_{mk} \alpha_{mi} \lvert \mathbf{\dot h}_{mk}^T \mathbf{\dot h}_{mi}^* \rvert ^2 + \alpha_{mi}\beta_{mk} \lvert \mathbf{\bar h}_{mk}^T \mathbf{\dot h}_{mi} ^*\rvert ^2 
  + \alpha_{mk}\alpha_{mi} \sqrt{\beta_{mk}} 2\Re( \mathbf{\dot h}_{mk}^T \mathbf{\dot h}_{mi}^* \mathbf{\dot h}_{mi}^T \mathbf{\bar h}_{mk}^*)  \right]
 \\ = q_{mk} q_{mi} \lvert \zeta_{mki} \rvert ^2 + \beta_{mk} q_{mi} \zeta_{mii}.
\end{multline*}
This expression simplifies to~\eqref{eq:stat_ki1} and
\eqref{eq:stat_ki2}. Since the definition of $\dot \gamma_{kk} $ remains unchanged, its statistics calculated in Appendix B still apply. Finally, 
\begin{equation*}
\begin{split}
   \E \left[ \lvert \bar\gamma_{kk}   \rvert^2 \right] = \E \left[\lvert \gamma_{kk} - \dot \gamma_{kk}\rvert ^2 \right] =    \E \left[\left\lvert\sum_{m=1} ^M x_{mk} \alpha_{mk} \mathbf{\bar h}_{mk}^T \mathbf{\dot h}_{mk}^*\right\rvert^2\right]
\end{split}
\end{equation*}
\begin{align}
    \E \left[ \lvert \bar\gamma_{kk}   \rvert^2 \right] &=\sum_{m=1}^M \lvert x_{mk}\rvert^2\E \left[\lvert \alpha_{mk} \mathbf{\bar h}_{mk}^T \mathbf{\dot h}_{mk}^* \rvert ^2\right] 
    \\& \nonumber+ \sum_{m=1}^M \sum_{\substack{m'=1\\m'\neq m}}^M x_{mk} x_{m'k}  \E \left[\alpha_{mk} \mathbf{\bar h}_{mk}^T  \mathbf{\dot h}_{mk}^*\right] \E \left[ \alpha_{m'k} \mathbf{\bar h}_{m'k}^T  \mathbf{\dot h}_{m'k}^*\right]^*
\end{align}
where, $ \E \left[\alpha_{mk} \mathbf{\bar h}_{mk}^T  \mathbf{\dot h}_{mk}^*\right]=0$ and,
$ \E \left[\lvert \alpha_{mk} \mathbf{\bar h}_{mk}^T \mathbf{\dot h}_{mk}^* \rvert ^2\right] = \beta_{mk} q_{mk} \zeta_{mkk}.   $ Therefore,
\begin{equation}
    \E \left[ \lvert \bar\gamma_{kk}   \rvert^2 \right] =\sum_{m=1}^M \lvert x_{mk}\rvert^2\beta_{mk} q_{mk} \zeta_{mkk}
\end{equation}
It is important to note that the two components $\dot\gamma_{kk}$ and
$\bar\gamma_{kk}$ are uncorrelated. Therefore,
\begin{equation}
    \E \left[ \lvert \gamma_{kk} \rvert ^2\right] = \E \left[ \lvert \dot\gamma_{kk} \rvert ^2\right] + \E \left[ \lvert \bar\gamma_{kk} \rvert ^2\right]
\end{equation}
\end{appendix}



%
\bibliographystyle{IEEEtran}
\bibliography{bibJournalList,main}

\end{document}